\documentclass[a4paper,11pt,showpacs,amsmath,amssymb,floatfix]{article}
\usepackage[utf8]{inputenc}
\usepackage{amsmath}
\usepackage{mathrsfs}
\usepackage{amsfonts}
 \usepackage{cancel}
\usepackage{graphicx}
\usepackage{amsthm}
\usepackage{enumerate}
\usepackage{amssymb}
\usepackage{multirow}
\usepackage{cite}
\usepackage{bm}
\usepackage{makecell}
\usepackage{float}
\usepackage{hyperref}
 \usepackage{setspace}
\usepackage{xcolor}
 	\definecolor{carblue}{rgb}{0.2, 0.30, 5}

\usepackage[colorinlistoftodos]{todonotes}

\usepackage{tikz}
\usetikzlibrary{patterns,calc}
\usetikzlibrary{patterns.meta}
\usetikzlibrary{arrows.meta,calc,decorations.pathreplacing,decorations.markings}
\usepackage{subcaption}

\usepackage{mathtools}
\mathtoolsset{showonlyrefs=true}

\newcommand{\lr}[1]{\left( #1 \right)}
\newcommand{\lrbrace}[1]{\left\lbrace #1 \right\rbrace}

\newcommand{\Om}{\Omega}

\newcommand{\scri}{\mathscr{I}}

\newcommand{\skwend}[1]{\mathrm{SkewEnd}(#1)}

\newcommand{\zeroset}[1]{\mathcal{Z}(#1)}
\newcommand{\closure}[1]{\overline{#1}}

\newcommand{\man}[1]{M^{#1}}
\newcommand{\phman}[1]{\tilde{M}^{#1}}

\newcommand{\conf}[1]{\mathrm{Conf}\left({#1}\right)}

\def\R{\mathbb{R}}
\def\S{\mathbb{S}}

\def\a{a}
\def\b{b}

\def\aa{\sigma}
\def\bb{\tau}
\def\cc{\sigma}

\def\dvg{\mathrm{div}}

\def\skwxi{w}
\def\sgrav{\mathcal{G}}

\def\Ric{\mathrm{Ric}}
\def\tq{~\slash~}

\newtheorem{theorem}{Theorem}[section]
\newtheorem{proposition}{Proposition}[section]
\newtheorem{corollary}{Corollary}[section]
\newtheorem{lemma}{Lemma}[section]
\newtheorem{remark}{Remark}[section]
\newtheorem{definition}{Definition}[section]

 \title{ Asymptotics of Killing horizons: A Characterization in de Sitter Spacetime}

\author{Carlos Pe\'on-Nieto\\
  Departamento de Matemática Aplicada a las TIC, \\
Universidad Politécnica de Madrid\\
{\small Calle Alan Turing s/n 28031, Madrid, Spain}}

\begin{document}

\maketitle

\begin{abstract}
We prove that $(n+1)$-dimensional de Sitter spacetime is locally the unique $\Lambda > 0$ vacuum spacetime whose Killing horizons intersect the conformal boundary $\mathscr{I}$. The intersection occurs at the isolated essential zeros of the extension of the Killing vector field to $\mathscr{I}$. We then reconstruct all de Sitter Killing vector fields from their boundary value and obtain a characterization of those admitting a Killing horizon in terms of a conformal invariant $\sigma$: $\sigma < 0$ characterizes non-degenerate bifurcate horizons, whereas $\sigma = 0$ yields degenerate non-bifurcate horizons. Finally, we derive an intrinsic, conformally invariant formula for the surface gravity of the horizons in de Sitter, $\kappa^2 = |\sigma|$.
\end{abstract}

 \section{Introduction}

 Killing horizons are an extremely non-generic feature of general relativistic spacetimes. This is true to such an extent that the presence of a Killing horizon within a $(\Lambda = 0)$-vacuum, asymptotically flat and stationary (or static) spacetime, one recovers a member of the Kerr family. This was established in  the celebrated uniqueness theorems, in the static case, by Israel, and in the stationary case, under additional technical assumptions, by Carter, Hawking, and Robinson (see, e.g., \cite{chruscielcostaheusler} for a comprehensive review with references). Moreover, removing the requirement of asymptotic flatness in the static case gives rise to only a few additional families \cite{reiris1,reiris2}.

Introducing a positive cosmological constant ($\Lambda > 0$) significantly complicates the analysis of Killing horizons. Among other reasons, this is mainly because the global structure of these spacetimes differs drastically from the $\Lambda = 0$ case. Indeed, in this context there exists an additional class of Killing horizons: the cosmological horizons. Physically, these horizons are explained by the repulsive effect of $\Lambda$, which bounds the causally accessible region for any given observer. Interestingly, in the ($\Lambda>0$)-vacuum case, cosmological horizons are typically accompanied by black hole horizons, for example, in the Schwarzschild-de Sitter and Kerr-de Sitter families. As a consequence, one expects cosmological horizons to be an equally non-generic feature of spacetimes with a positive cosmological constant. For example, in recent works that explore uniqueness of Kerr-de Sitter, by Hintz \cite{hintz18}, Hintz-Vasy \cite{hintzvasy18} and Fang \cite{fang}, the cosmological horizon and black hole horizon play an analogue role.

Although cosmological horizons have traditionally been studied through their global spacetime geometry, much less is known about their characterization directly from the asymptotic geometry at conformal infinity, $\scri$. A step in that direction is taken in \cite{kaminskisymmetries}. There it was shown that four-dimensional de Sitter spacetime is locally characterized as the unique spacetime where cosmological horizons intersect the conformal boundary $\scri$.
More specifically, the generator of the horizons extends to $\scri$ as a tangential conformal Killing vector field thereof and the intersection  happens at its \emph{isolated essential zeros}, which is a subset of its isolated zeros defined by the local conformal properties of the vector field (cf. Definition \ref{defessential}). In this paper (cf. Section \ref{secnecessary}), we extend this result to arbitrary dimensions.

While this provides a local characterization of de Sitter as the unique spacetime with a Killing horizon intersecting $\scri$, the asymptotic characterization and classification of its Killing horizons has yet to be completed. Namely, in de Sitter spacetime, one must determine which  conformal Killing vector fields of $\scri$ with essential zeros correspond to Killing vector fields admitting a Killing horizon. We solve this problem in Section \ref{secsufficient} by directly reconstructing the Killing vector fields from their boundary value at $\scri$ and then evaluating the conditions under which they admit a Killing horizon. This completes the asymptotic characterization and classification of Killing horizons in de Sitter spacetime. Additionally, in Section \ref{secgrav} we derive a formula to compute the surface gravity of these  horizons at its intersection with $\scri$, which is given by a set of conformal invariants derived from the  conformal Killing vector fields.

The present paper focuses exclusively on de Sitter spacetime. However, these results provide an essential building block toward a general framework for the asymptotic analysis of Killing horizons. Roughly speaking, de Sitter spacetime serves as a reference \emph{background} in the asymptotic region for spacetimes with conformally flat $\scri$. More precisely, on de Sitter's conformal boundary, one can define more general asymptotic data and analyze its local evolution within the background de Sitter's topology. With respect to the background topology, one may then take limits beyond the domain of the more general spacetime to study its behavior near the limit. In particular, we aim to understand whether the presence of cosmological Killing horizons (which in many relevant spacetimes coincides with the Cauchy horizons of the domain of dependence of $\scri$) manifests as a divergent behavior of the asymptotic data. This analysis is left for a future publication, although we shall comment further on these ideas in Section \ref{secdiscussion}.

For the reader's convenience, in Section~\ref{secresults} we include an overview of the main results of this paper. Section~\ref{sectools} describes the mathematical framework and background tools required throughout. As previously mentioned, Section~\ref{secnecessary} presents a detailed (local) characterization of de Sitter spacetime as the unique spacetime whose Killing horizons intersect $\scri$ in a non-empty set. Section~\ref{secsufficient} provides a characterization and classification of these horizons at $\scri$, while Section~\ref{secgrav} derives a formula for computing their surface gravity. Future applications of these results are discussed in Section~\ref{secdiscussion}. Finally, Appendix~\ref{appcanonical} presents the details of a classification of conformal Killing vector fields modulo conformal isometries, required for Section \ref{secsufficient}, and Appendix~\ref{appdF} collects several elementary but lengthy computations.

 \section{Main results}\label{secresults}

 The first main result of this paper (Theorem \ref{theoremnecessary}) addresses the question of which $(\Lambda>0)$-vacuum spacetimes with a Killing vector field (for brevity, Killing vector) $X$ admitting a Killing horizon $\mathcal{H}_X$ satisfy that its closure, denoted $\closure{\mathcal{H}_X}$, intersects $\scri$. As explained in the introduction, this extends the four-dimensional characterization obtained in \cite{kaminskisymmetries} to arbitrary dimensions.

The initial observation is that $X$ always extends smoothly to any connected component of null infinity (say, $\scri^+$) as a tangential conformal Killing vector field (CKV) $\xi \coloneqq X|_{\scri^+}$ (see e.g. \cite{marspeondata21}). Since $X$ is null and tangent to $\mathcal{H}_X$, the closure of the horizon can intersect $\scri^+$ only at the zeros of $\xi$. A more subtle argument (Proposition \ref{propisolatedzeros}) shows that these zeros must, in fact, be \emph{essential isolated zeros} (Definition \ref{defessential}), that is, zeros where $\xi$ cannot be made Killing under any local conformal rescaling of the metric. Otherwise, $X$ would necessarily be spacelike in a neighbourhood of $\scri^+$, preventing from the existence of a Killing horizon arbitrarily close to $\scri^+$.

The final ingredient is the Obata-Ferrand theorem \cite{obata,obata0,ferrand,schoen}, in its infinitesimal version due to Frances \cite{frances} and Belgun {\it et al} \cite{belgunmoroianuornea}. These imply that a Riemannian manifold admitting a conformal Killing vector with essential isolated zeros is locally conformally flat near these zeros. Integrating the Killing initial data equation (Theorem \ref{theoKIDanal}) then shows that the asymptotic data coincide with that of de Sitter in a neighbourhood of the essential zeros. If $\xi$ is complete, its integral curves cover the whole conformal boundary, so that the Killing initial data equation can be globally integrated and the conclusion is global.

The converse statement is straightforward from some of the arguments above together with basic properties of de Sitter spacetime (see Remark \ref{remarksuff}). This leads to the following characterization Theorem.

\begin{theorem}\label{theo1}
  Let $(\tilde M^{n+1},\tilde g)$ be a $(\Lambda>0)$-vacuum spacetime admitting a conformal extension and a Killing vector field $X$, with $\xi = X\mid_{\scri^+}$. Then $X$ has a Killing horizon $\mathcal{H}_X$ whose closure $\closure{\mathcal{H}_X}$ satisfies $\closure{\mathcal{H}_X}\cap \scri^+ \neq \emptyset$ only if $\xi$ has essential isolated zeros at $\closure{\mathcal{H}_X}\cap \scri^+$ and  $(\tilde M^{n+1},\tilde g)$ is isometric to de Sitter in a neighbourhood of $\closure{\mathcal{H}_X}\cap \scri^+$.  Moreover if $\xi$ is complete, the result applies in a neighbourhood of $\scri^+$.

  Conversely, all Killing vectors $X$ of de Sitter spacetime admitting a Killing horizon $\mathcal{H}_X$ satisfy that $\closure{\mathcal{H}_X}\cap \scri^+ \neq \emptyset$ and it coincides with the essential isolated zeros of $\xi = X\mid_{\scri^+}$.
\end{theorem}

Note that Theorem \ref{theo1} also entails a partial characterization of de Sitter  Killing horizons at $\scri^+$. Namely, a Killing vector of de Sitter $X$ admits a Killing horizon only if $\xi = X\mid_{\scri^+}$ has essential isolated zeros, which moreover coincides with $\closure{\mathcal{H}_X}\cap \scri^+$. In order to complete the characterization, it remains to decide which essential zeros of CKVs of $\scri^+$ do correspond with the intersection $\closure{\mathcal{H}_X}\cap \scri^+$ and which ones  do not. We solve this problem in Section \ref{secsufficient}.

The approach of Section \ref{secsufficient} is constructive. Namely, we reconstruct the Killing vector field from its boundary value, i.e. the CKV $\xi$, and determine whether it admits a Killing horizon. The conformal invariance of the asymptotic initial value problem, together with the maximal symmetry of de Sitter, allows us to avoid a case-by-case analysis and solve the problem at once.

First, the conformal invariance allows us to choose the most convenient conformal extension of de Sitter.  On the other hand, in de Sitter, the conformal group of (conformally flat) $\scri^+$ is in one-to-one correspondence with group of orthochronous isometries of the bulk spacetime. Thus, CKVs of $\scri^+$ related by a conformal transformation (that is, belonging to the same conformal class) extend to Killing vectors related by an isometry. Hence, this gives us freedom to equivalently select any CKV of any given conformal class.

Using the results of \cite{marspeon21.1,marspeon21} (reviewed in Appendix \ref{appcanonical}), we can handle all the conformal classes of CKVs at once by means of its canonical representative $\xi_{\mathrm{can}}$. This is a CKV depending only on a number of conformal invariants and representing the entire space of conformal classes. Extending $\xi_{\mathrm{can}}$ to the corresponding Killing vector $X_{\mathrm{can}}$, which is a canonical representative of its isometry class, we carry out the analysis in a unified way.

The result is, as expected, that the CKV must have isolated essential zeros and also a collection of \emph{rotation parameters} (see Appendix \ref{appcanonical}) identifying its conformal class must vanish. The isolated zeros coincide with the intersection points of the horizons with $\scri^+$. Since $\scri^+$ is locally conformally flat, $\xi_{\mathrm{can}}$ has either one or two isolated zeros. The case of two isolated zeros corresponds to a bifurcate horizon, while a single isolated zero gives rise to a non-bifurcate, and hence degenerate, horizon.

Interestingly, in the bifurcate case, the ``distance'' between the two intersection points, as measured with respect to a fixed conformal representative of $\scri^+$,  is determined by one of the conformal invariants  defining $\xi_{\mathrm{can}}$, denoted by $\sigma$. As $\sigma\to0$, the two horizons approach each other and merge into a single horizon in the limit, recovering the degenerate non-bifurcate case.

This inspires the analysis carried out in Section \ref{secgrav}, where we show that this parameter determines the surface gravity. More precisely, we derive a formula expressing the surface gravity entirely in terms of intrinsic quantities associated with $\xi$.

\bigskip

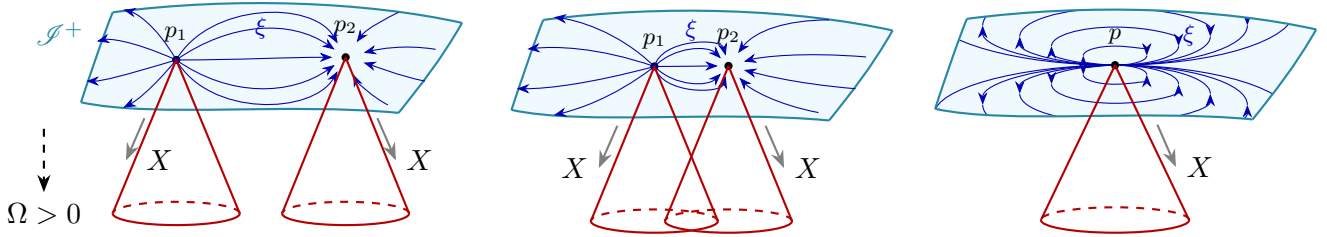
\begin{figure}[htbp]
\centerline{%
        \begin{minipage}{1.1\textwidth}
\begin{subfigure}[b]{0.32\textwidth}
        \centering
\begin{tikzpicture}[
    >=Stealth,
    thick,
    scale=0.7
]

\fill[cyan!5] (1.2, 0.8)
    .. controls (3.5, 1.1) and (6.0, 1.0) .. (7.8, 0.6)   
    .. controls (7.4, -0.1) and (7.0, -0.6) .. (6.6, -1.1) 
    .. controls (4.5, -0.9) and (2.0, -1.2) .. (0.6, -0.9) 
    .. controls (0.8, -0.3) and (1.0, 0.3) .. (1.2, 0.8);  

\draw[cyan!60!black] (1.2, 0.8) .. controls (3.5, 1.1) and (6.0, 1.0) .. (7.8, 0.6);
\draw[cyan!60!black] (7.8, 0.6) .. controls (7.4, -0.1) and (7.0, -0.6) .. (6.6, -1.1);
\draw[cyan!60!black] (6.6, -1.1) .. controls (4.5, -0.9) and (2.0, -1.2) .. (0.6, -0.9);
\draw[cyan!60!black] (0.6, -0.9) .. controls (0.8, -0.3) and (1.0, 0.3) .. (1.2, 0.8);

\node[cyan!60!black] at (0.2, 0.4) {$\scri^+$};

\coordinate (L) at (2.4, -0.1);
\coordinate (R) at (5.6, -0.05);

\fill[black] (L) circle (2.2pt) node[above, shift={(0., 0.1)}] {\footnotesize $p_1$};
\fill[black] (R) circle (2.2pt) node[above, shift={(0., 0.15)}] {\footnotesize $p_2$};


\draw[->, blue!70!black, thin, shorten >= 5pt] (L) to[out=35, in=145] (R);
\draw[->, blue!70!black, thin, shorten >= 5pt] (L) to[out=0, in=180] (R);
\draw[->, blue!70!black, thin, shorten >= 5pt] (L) to[out=-35, in=-145] (R);

\draw[->, blue!70!black, thin, shorten >= 6pt] (L) to[out=60, in=120] (R);
\draw[->, blue!70!black, thin, shorten >= 6pt] (L) to[out=-60, in=-120] (R);

\draw[->, blue!70!black, thin] (L) .. controls (2.0, 0.6)  and (1.6, 0.8)  .. (1.4, 0.85);
\draw[->, blue!70!black, thin] (L) .. controls (1.5, 0.3)  and (1.1, 0.3)  .. (0.9, 0.25);
\draw[->, blue!70!black, thin] (L) .. controls (1.4, -0.4) and (1.0, -0.5) .. (0.7, -0.5);
\draw[->, blue!70!black, thin] (L) .. controls (1.9, -0.7) and (1.6, -0.9) .. (1.4, -1.0);

\draw[->, blue!70!black, thin, shorten >= 7pt] (6.6, 0.8)   .. controls (6.4, 0.7)  and (6.0, 0.5)  .. (R);
\draw[->, blue!70!black, thin, shorten >= 7pt] (7.3, 0.1)   .. controls (6.7, 0.2)  and (6.2, 0.1)  .. (R);
\draw[->, blue!70!black, thin, shorten >= 7pt] (7.1, -0.5)  .. controls (6.6, -0.4) and (6.1, -0.2) .. (R);
\draw[->, blue!70!black, thin, shorten >= 7pt] (6.4, -0.95) .. controls (6.1, -0.8) and (5.8, -0.5) .. (R);

\node[blue!80!black] at (4.0, 0.5) {\footnotesize $\xi$};

%


\draw[red!70!black, dashed] (3.6, -3.0) arc (0:180:1.2cm and 0.22cm);

\draw[red!70!black] (L) -- (1.2, -3.0);
\draw[red!70!black] (L) -- (3.6, -3.0);

\draw[red!70!black] (1.2, -3.0) arc (180:360:1.2cm and 0.22cm);

\draw[->, gray] (1.8, -1.2) -- ++(-0.35,-0.8) node[right, black] {$~X$};


\draw[red!70!black, dashed] (6.8, -3.0) arc (0:180:1.2cm and 0.22cm);

\draw[red!70!black] (R) -- (4.4, -3.0);
\draw[red!70!black] (R) -- (6.8, -3.0);

\draw[red!70!black] (4.4, -3.0) arc (180:360:1.2cm and 0.22cm);

\draw[->, gray] (6.2, -1.2) -- ++(0.35,-0.8) node[right, black] {$X$};

\draw[->, black, dashed] (-0.1, -1.4) -- (-0.1, -2.6) node[below] {$\Omega>0$ };

\end{tikzpicture}
\end{subfigure}
\hfill\hspace{10pt}
\begin{subfigure}[b]{0.32\textwidth}
        \centering
\begin{tikzpicture}[>=Stealth, thick, scale=0.7]

\fill[cyan!5] (1.2, 0.8)
    .. controls (3.5, 1.1) and (6.0, 1.0) .. (7.8, 0.6)
    .. controls (7.4, -0.1) and (7.0, -0.6) .. (6.6, -1.1)
    .. controls (4.5, -0.9) and (2.0, -1.2) .. (0.6, -0.9)
    .. controls (0.8, -0.3) and (1.0, 0.3) .. (1.2, 0.8);

\draw[cyan!60!black] (1.2, 0.8) .. controls (3.5, 1.1) and (6.0, 1.0) .. (7.8, 0.6);
\draw[cyan!60!black] (7.8, 0.6) .. controls (7.4, -0.1) and (7.0, -0.6) .. (6.6, -1.1);
\draw[cyan!60!black] (6.6, -1.1) .. controls (4.5, -0.9) and (2.0, -1.2) .. (0.6, -0.9);
\draw[cyan!60!black] (0.6, -0.9) .. controls (0.8, -0.3) and (1.0, 0.3) .. (1.2, 0.8);


\coordinate (L) at (3.3, -0.08);
\coordinate (R) at (4.7, -0.07);
\fill[black] (L) circle (2.2pt) node[above, shift={(0., 0.1)}] {\footnotesize $p_1$};
\fill[black] (R) circle (2.2pt) node[above, shift={(0., 0.15)}] {\footnotesize $p_2$};

\draw[->, blue!70!black, thin, shorten >= 4pt] (L) to[out=45, in=135] (R);
\draw[->, blue!70!black, thin, shorten >= 4pt] (L) to[out=0, in=180] (R);
\draw[->, blue!70!black, thin, shorten >= 4pt] (L) to[out=-45, in=-135] (R);
\draw[->, blue!70!black, thin, shorten >= 4pt] (L) to[out=75, in=105] (R);
\draw[->, blue!70!black, thin, shorten >= 4pt] (L) to[out=-75, in=-105] (R);

\draw[->, blue!70!black, thin] (L) .. controls (2.6, 0.6)  and (1.8, 0.8)  .. (1.4, 0.85);
\draw[->, blue!70!black, thin] (L) .. controls (2.2, 0.3)  and (1.4, 0.3)  .. (0.9, 0.25);
\draw[->, blue!70!black, thin] (L) .. controls (2.1, -0.4) and (1.3, -0.5) .. (0.7, -0.5);
\draw[->, blue!70!black, thin] (L) .. controls (2.5, -0.7) and (1.8, -0.9) .. (1.4, -1.0);

\draw[->, blue!70!black, thin, shorten >= 6pt] (6.6, 0.8)   .. controls (5.8, 0.8)  and (5.2, 0.6)  .. (R);
\draw[->, blue!70!black, thin, shorten >= 6pt] (7.3, 0.1)   .. controls (6.1, 0.2)  and (5.3, 0.1)  .. (R);
\draw[->, blue!70!black, thin, shorten >= 6pt] (7.1, -0.5)  .. controls (5.9, -0.4) and (5.2, -0.2) .. (R);
\draw[->, blue!70!black, thin, shorten >= 6pt] (6.4, -0.95) .. controls (5.6, -0.8) and (5.1, -0.5) .. (R);

\node[blue!80!black] at (4.0, 0.6) {\footnotesize $\xi$};

\draw[red!70!black, dashed] (4.5, -3.0) arc (0:180:1.2cm and 0.22cm);
\draw[red!70!black] (L) -- (2.1, -3.0);
\draw[red!70!black] (L) -- (4.5, -3.0);
\draw[red!70!black] (2.1, -3.0) arc (180:360:1.2cm and 0.22cm);
\draw[->, gray] (2.6, -1.2) -- ++(-0.35,-0.8) node[left, black] {$X$};

\draw[red!70!black, dashed] (5.9, -3.0) arc (0:180:1.2cm and 0.22cm);
\draw[red!70!black] (R) -- (3.5, -3.0);
\draw[red!70!black] (R) -- (5.9, -3.0);
\draw[red!70!black] (3.5, -3.0) arc (180:360:1.2cm and 0.22cm);
\draw[->, gray] (5.4, -1.2) -- ++(0.35,-0.8) node[right, black] {$X$};
\end{tikzpicture}
 \end{subfigure}
 \hfill\hspace{-15pt}
\begin{subfigure}[b]{0.32\textwidth}
        \centering
\begin{tikzpicture}[>=Stealth, thick, scale=0.7]

\fill[cyan!5] (1.2, 0.8)
    .. controls (3.5, 1.1) and (6.0, 1.0) .. (7.8, 0.6)
    .. controls (7.4, -0.1) and (7.0, -0.6) .. (6.6, -1.1)
    .. controls (4.5, -0.9) and (2.0, -1.2) .. (0.6, -0.9)
    .. controls (0.8, -0.3) and (1.0, 0.3) .. (1.2, 0.8);

\begin{scope}
    \clip (1.2, 0.8)
        .. controls (3.5, 1.1) and (6.0, 1.0) .. (7.8, 0.6)
        .. controls (7.4, -0.1) and (7.0, -0.6) .. (6.6, -1.1)
        .. controls (4.5, -0.9) and (2.0, -1.2) .. (0.6, -0.9)
        .. controls (0.8, -0.3) and (1.0, 0.3) .. (1.2, 0.8);

    \coordinate (M) at (4.0, -0.075);


    \draw[->, blue!70!black, thin] (4.0, -0.075) arc [start angle=270, end angle=180, x radius=0.6, y radius=0.18];
    \draw[blue!70!black, thin] (3.4, 0.105) arc [start angle=180, end angle=90, x radius=0.6, y radius=0.18];
    \draw[->, blue!70!black, thin] (4.0, 0.285) arc [start angle=90, end angle=0, x radius=0.6, y radius=0.18];
    \draw[blue!70!black, thin] (4.6, 0.105) arc [start angle=0, end angle=-90, x radius=0.6, y radius=0.18];

    \draw[->, blue!70!black, thin] (4.0, -0.075) arc [start angle=270, end angle=180, x radius=1.2, y radius=0.36];
    \draw[blue!70!black, thin] (2.8, 0.285) arc [start angle=180, end angle=90, x radius=1.2, y radius=0.36];
    \draw[->, blue!70!black, thin] (4.0, 0.645) arc [start angle=90, end angle=0, x radius=1.2, y radius=0.36];
    \draw[blue!70!black, thin] (5.2, 0.285) arc [start angle=0, end angle=-90, x radius=1.2, y radius=0.36];

    \draw[->, blue!70!black, thin] (4.0, -0.075) arc [start angle=270, end angle=180, x radius=1.8, y radius=0.54];
    \draw[blue!70!black, thin] (2.2, 0.465) arc [start angle=180, end angle=90, x radius=1.8, y radius=0.54];
    \draw[->, blue!70!black, thin] (4.0, 1.005) arc [start angle=90, end angle=0, x radius=1.8, y radius=0.54];
    \draw[blue!70!black, thin] (5.8, 0.465) arc [start angle=0, end angle=-90, x radius=1.8, y radius=0.54];

    \draw[->, blue!70!black, thin] (4.0, -0.075) arc [start angle=270, end angle=180, x radius=2.5, y radius=0.75];
    \draw[blue!70!black, thin] (1.5, 0.675) arc [start angle=180, end angle=90, x radius=2.5, y radius=0.75];
    \draw[->, blue!70!black, thin] (4.0, 1.425) arc [start angle=90, end angle=0, x radius=2.5, y radius=0.75];
    \draw[blue!70!black, thin] (6.5, 0.675) arc [start angle=0, end angle=-90, x radius=2.5, y radius=0.75];

    \draw[->, blue!70!black, thin] (4.0, -0.075) arc [start angle=270, end angle=180, x radius=3.4, y radius=1.05];
    \draw[blue!70!black, thin] (0.6, 0.975) arc [start angle=180, end angle=90, x radius=3.4, y radius=1.05];
    \draw[->, blue!70!black, thin] (4.0, 2.025) arc [start angle=90, end angle=0, x radius=3.4, y radius=1.05];
    \draw[blue!70!black, thin] (7.4, 0.975) arc [start angle=0, end angle=-90, x radius=3.4, y radius=1.05];

    \draw[->, blue!70!black, thin] (4.0, -0.075) arc [start angle=90, end angle=180, x radius=0.6, y radius=0.18];
    \draw[blue!70!black, thin] (3.4, -0.255) arc [start angle=180, end angle=270, x radius=0.6, y radius=0.18];
    \draw[->, blue!70!black, thin] (4.0, -0.435) arc [start angle=270, end angle=360, x radius=0.6, y radius=0.18];
    \draw[blue!70!black, thin] (4.6, -0.255) arc [start angle=0, end angle=90, x radius=0.6, y radius=0.18];

    \draw[->, blue!70!black, thin] (4.0, -0.075) arc [start angle=90, end angle=180, x radius=1.2, y radius=0.36];
    \draw[blue!70!black, thin] (2.8, -0.435) arc [start angle=180, end angle=270, x radius=1.2, y radius=0.36];
    \draw[->, blue!70!black, thin] (4.0, -0.795) arc [start angle=270, end angle=360, x radius=1.2, y radius=0.36];
    \draw[blue!70!black, thin] (5.2, -0.435) arc [start angle=0, end angle=90, x radius=1.2, y radius=0.36];

    \draw[->, blue!70!black, thin] (4.0, -0.075) arc [start angle=90, end angle=180, x radius=1.8, y radius=0.54];
    \draw[blue!70!black, thin] (2.2, -0.615) arc [start angle=180, end angle=270, x radius=1.8, y radius=0.54];
    \draw[->, blue!70!black, thin] (4.0, -1.155) arc [start angle=270, end angle=360, x radius=1.8, y radius=0.54];
    \draw[blue!70!black, thin] (5.8, -0.615) arc [start angle=0, end angle=90, x radius=1.8, y radius=0.54];

    \draw[->, blue!70!black, thin] (4.0, -0.075) arc [start angle=90, end angle=180, x radius=2.5, y radius=0.75];
    \draw[blue!70!black, thin] (1.5, -0.825) arc [start angle=180, end angle=270, x radius=2.5, y radius=0.75];
    \draw[->, blue!70!black, thin] (4.0, -1.575) arc [start angle=270, end angle=360, x radius=2.5, y radius=0.75];
    \draw[blue!70!black, thin] (6.5, -0.825) arc [start angle=0, end angle=90, x radius=2.5, y radius=0.75];

    \draw[->, blue!70!black, thin] (4.0, -0.075) arc [start angle=90, end angle=180, x radius=3.4, y radius=1.05];
    \draw[blue!70!black, thin] (0.6, -1.125) arc [start angle=180, end angle=270, x radius=3.4, y radius=1.05];
    \draw[->, blue!70!black, thin] (4.0, -2.175) arc [start angle=270, end angle=360, x radius=3.4, y radius=1.05];
    \draw[blue!70!black, thin] (7.4, -1.125) arc [start angle=0, end angle=90, x radius=3.4, y radius=1.05];

    \node[blue!80!black] at (5.4, 0.5) {\footnotesize $\xi$};
\end{scope}

\draw[cyan!60!black] (1.2, 0.8) .. controls (3.5, 1.1) and (6.0, 1.0) .. (7.8, 0.6);
\draw[cyan!60!black] (7.8, 0.6) .. controls (7.4, -0.1) and (7.0, -0.6) .. (6.6, -1.1);
\draw[cyan!60!black] (6.6, -1.1) .. controls (4.5, -0.9) and (2.0, -1.2) .. (0.6, -0.9);
\draw[cyan!60!black] (0.6, -0.9) .. controls (0.8, -0.3) and (1.0, 0.3) .. (1.2, 0.8);


\fill[black] (M) circle (2.2pt)node[above=4pt] {\footnotesize $p$};

\draw[red!70!black, dashed] (5.4, -3.0) arc (0:180:1.4cm and 0.24cm);
\draw[red!70!black] (M) -- (2.6, -3.0);
\draw[red!70!black] (M) -- (5.4, -3.0);
\draw[red!70!black] (2.6, -3.0) arc (180:360:1.4cm and 0.24cm);

\draw[->, gray] (4.8, -1.2) -- ++(0.35,-0.8) node[right, black] {$X$};

\end{tikzpicture}
\end{subfigure}
\end{minipage}%
    }

\caption{Horizons intersecting $\scri^+$ at the isolated zeros of $\xi$. As the surface gravity reduces, the bifurcate horizon branches merge into one degenerate horizon.}
\end{figure}

\bigskip

In this introduction, we present a slightly less technical version of Theorem \ref{theosuff}, in which the conditions for a CKV $\xi$ to extend to a Killing vector admitting a horizon are formulated geometrically. We use a representative of the conformal class $[\xi]$ on the flat Euclidean space $\mathbb{E}^n$. Since $\mathbb{E}^n$ is only a local conformal representation of the global sphere $\mathbb{S}^n$, some of the isolated zeros of the CKVs may occur at infinity in the Euclidean picture. Therefore, the theorem should be understood on the one-point compactification $\S^n = \mathbb{E}^n \cup \{ \infty \}$.

 \begin{theorem}\label{theo2}
  Consider $(n+1)$-dimensional de Sitter spacetime $(\tilde M,\tilde g)$ and let $X$ be a Killing vector thereof, where $\xi = X\mid_{\scri^+}$. Then
  \begin{enumerate}[a)]
   \item    $X$ admits a bifurcate, thus non-degenerate, horizon $\mathcal{H}_X$ if and only if  $\xi$ is conformally equivalent to a homothety, thus has exactly two isolated zeros (one of them possibly at infinity in $\S^n = \mathbb{E}^n \cup \{ \infty \}$), which coincide with $\closure{\mathcal{H}_X}\cap \scri^+$.
     \item    $X$ admits a non-bifurcate, thus degenerate, horizon $\mathcal{H}_X$ if and only if  $\xi$ is conformally equivalent to a translation, thus has exactly one isolated zero (possibly at infinity in $\S^n = \mathbb{E}^n \cup \{ \infty \}$), which coincides with $\closure{\mathcal{H}_X}\cap \scri^+$.
  \end{enumerate}
 \end{theorem}

 \section {Generalities and tools}\label{sectools}

We start by reviewing general results and considerations. We will consider $\Lambda >0$ vacuum spacetimes of $(n+1)$ dimensions, with $n \geq 3$. Namely, Lorentzian manifolds $(\tilde M^{n+1}, \tilde g)$, satisfying the Einstein equation
 \begin{equation}\label{Einstein}
\tilde \Ric = n \lambda \tilde g, \qquad \lambda >0.
\end{equation}
The constant $\lambda$ is a reduced cosmological constant, related to the usual one by $\lambda = \frac{\Lambda}{n(n-1)}$.

Abstract indices in the main text are the following. Lower case latin indices $a,b,c, \cdots$ are used for ($n+1$)-dimensional Lorentzian spacetimes, upper case latin indices $A,B,C, \cdots$ are for its embedded spacelike $n$-dimensional submanifolds and Greek indices $\alpha,\beta,\cdots$ are only used for an ($n+2$)-dimensional space described in Appendix \ref{appcanonical}.

\subsection{Killing horizons}\label{secdefHx}

Given a Killing vector field $X$ of $\tilde g$, a Killing horizon $\mathcal{H}_X$ is an embedded, null, connected hypersurface such that $X$ is null and nowhere zero on $\mathcal{H}_X$, i.e.,
\begin{equation}
 \mathcal{H}_X =  \{ p \in \tilde M^{n+1}\tq \tilde g(X,X)(p)=0,~X_p \neq 0\}.
\end{equation}
Whenever $X$ vanishes on a co-dimension two connected submanifold $S$, the null hypersurfaces $\mathcal{H}_1$ and $\mathcal{H}_2$ generated by the two transversal null directions give raise to Killing horizons of $X$ on each  (future and past) connected component of $\mathcal{H}_1 ~ \backslash ~  S$ and $\mathcal{H}_2 ~ \backslash ~ S$. A bifurcate horizon is the union $\mathcal{H}_X = \mathcal{H}_1 \cup \mathcal{H}_1$.

\subsection{Conformal extensions}

We shall consider spacetimes $(\tilde M,\tilde g)$ admitting a conformal extension, i.e., an embedding into a  manifold with boundary $(M^{n+1},g)$ such that
\begin{equation}
 \phi: \tilde M^{n+1} \hookrightarrow M^{n+1},\qquad \phi^\star(\Om^{-2}g) =  \tilde g,
\end{equation}
for a smooth\footnote{Unless otherwise stated, for simplicity we assume $C^\infty$ smoothness throughout, although minimal regularity requirements may be lower.} function $\Om: M^{n+1} \to \mathbb{R}$ satisfying
\begin{equation}
 \phi(\phman{n+1}) = \{ p \in M^{n+1} \tq \Om(p) > 0 \}, \qquad  \scri := \partial M^{n+1} = \{ p \in M^{n+1} \tq \Om(p)=0 \mbox{ and } d_p \Om  \neq 0 \} .
\end{equation}
Identifying objects with their images in $\man{n+1}$, we write $g=\Om^2\tilde g$. In the remainder, objects with tilde are computed with $\tilde g$ while objects without tilde are derived from $g$. Their respective Levi-Civita connections are  $\tilde \nabla$ and $\nabla$ and the associated Riemann tensor, say for $\nabla$, follows the convention
\begin{equation}
 R^a{}_{bcd} X^b = \nabla_c \nabla_d X^a - \nabla_d \nabla_c X^a.
\end{equation}
The Ricci tensors of  $\tilde g$ and $g$ satisfy the relation
\begin{equation}
 \tilde R_{ab}-R_{ab} = \frac{n-1}{\Om} \nabla_a \nabla_b \Om + g_{ab} \frac{\Box \Om}{\Om} - g_{ab} \frac{n}{\Om^2} |\nabla \Om|^2,
\end{equation}
so whenever $\tilde g$ satisfies the Einstein equation \eqref{Einstein}, it induces the so-called quasi-Einstein equation on $g$
\begin{equation}\label{quasiEinstein}
R_{ab} = -\frac{n-1}{\Om} \nabla_a \nabla_b \Om  -g_{ab} \frac{\Box \Om}{\Om} + g_{ab} \frac{n}{\Om^2} \left( |\nabla \Om|^2 + \lambda \right).
\end{equation}
From \eqref{quasiEinstein}, it follows readily that $g$ induces an $n$-dimensional Riemannian metric on $\scri$, which generally has two connected components, $\scri=\scri^+\cup\scri^-$. Since conformal extensions are not unique, $\scri$ inherits a conformal structure, that is, a conformal class of metrics
\begin{equation}\label{confclasgamma}
[\gamma]={\gamma'=\omega^2\gamma,\ \forall\omega\in\mathcal{C}^\infty(\Sigma^n)},
\end{equation}
whose representatives are induced metrics on $\scri$ for suitable conformal extensions. We shall often regard $\scri$ as an embedded spacelike $n$-manifold, so that its intrinsic geometry may be studied in terms of an abstract embedded $n$-manifold $(\Sigma^n,\gamma)\hookrightarrow(\scri,\gamma)\subset(M^{n+1},g)$. Furthermore, whenever we derive local results near $\scri$,
whenever we formulate it only near the future connected component $\scri^+$, but they apply near $\scri^-$ too.

\subsection{Asymptotic value problem}\label{secdata}

The asymptotic value problem is an initial value problem for $g$ with data prescribed at $\scri^+$. For simplicity, we shall not describe its full construction and instead limit ourselves to a brief overview, highlighting the main references. Since our aim is an arbitrary-dimensional treatment, we shall not discuss specifically four-dimensional methods such as Friedrich's celebrated conformal field equations \cite{friedrich81,friedrich81bis,Fried86initvalue}. Nevertheless, it can be reformulated in terms of the framework discussed below (see, e.g.,  \cite{marspeonRSTA}).

\bigskip

The asymptotic degrees of freedom of $g$ become apparent through the Fefferman-Graham expansion \cite{FeffGrah85,ambientmetric}. Given any representative $\gamma\in[\gamma]$, one may choose a {\it geodesic conformal extension}, i.e. a conformal extension whose conformal factor satisfies
\begin{equation}\label{geodesicgauge}
 |\nabla\Omega|_g^2=-\lambda.
\end{equation}
In coordinates $\{\Omega,x^A\}$ adapted to the level sets of $\Omega$, the metric takes the form
\begin{equation}\label{eqggeodgauge}
 g=-\frac{d\Omega^2}{\lambda}+g_\Omega,
\end{equation}
where $g_{\Om}$ is a family of Riemanninan $n$-metrics induced in $\{ \Om = const.\}$.
In rough terms, the Fefferman-Graham expansion is an asymptotic expansion  of $g_\Omega$ near $\scri^+$, whose coefficients are recursively determined by equation \eqref{quasiEinstein}, which in this conformal gauge is
\begin{equation}\label{qEinsteingeodgauge}
R_{ab} = -\frac{n-1}{\Om} \nabla_a \nabla_b \Om  -g_{ab} \frac{\Box \Om}{\Om} .
\end{equation}

When $n$ is odd, the expansion looks
\begin{equation}\label{expnodd}
g_\Om =
g_{(0)}
+\Omega^2 g_{(2)}
+\cdots
+\Omega^{n-1}g_{(n-1)}
+\Omega^n g_{(n)}
+O(\Omega^{n+1}).
\end{equation}
It contains only even power terms of $\Omega$ up to the order $n-1$,
all of them uniquely determined by $g_{(0)}=\gamma$. The coefficient $g_{(n)}$ is undetermined by the recursion, except that
\begin{equation}\label{constgnodd}
 \mathrm{Tr}_\gamma g_{(n)} = 0,\qquad \mathrm{div}_\gamma g_{(n)} = 0.
\end{equation}
 Terms of order higher than $n$ are also powers of $\Om$ and are determined by $(g_{(0)},g_{(n)})$.

 When $n$ is even, the expansion
\begin{equation}\label{expneven}
g_\Om =
g_{(0)}
+\Omega^2 g_{(2)}
+\cdots
+\Omega^n g_{(n)}
+O(\Omega^{n}\log\Omega),
\end{equation}
 contains only even powers of $\Omega$ of order lower than $n$, but from the order $\Omega^{n}\log\Omega$ onwards it becomes polyhomogeneous, i.e. it includes terms of the form $\Om^{n+s}(\log \Om)^l$, $l \leq s$. The coefficients $g_{(k)}$ with $k<n$ are uniquely determined by $g_{(0)}=\gamma$, while $g_{(n)}$ is not determined by the recursion, except that
\begin{equation}\label{constgneven}
\mathrm{Tr}_\gamma g_{(n)}=\mathfrak{a}(\gamma),
\qquad
\mathrm{div}_\gamma g_{(n)}=\mathfrak{b}(\gamma),
\end{equation}
for scalar and tensor functions $\mathfrak{a}(\gamma)$ and $\mathfrak{b}(\gamma)$ respectively, in general non-vanishing. Higher-order terms are polyhomogeneous or even powers of  $\Omega$, all determined by $(g_{(0)},g_{(n)})$. Remarkably, the presence of logarithmic terms depends entirely on the so-called {\it obstruction tensor}, which is determined by $\gamma$. For instance, for $\gamma$  conformally flat, the obstruction tensor vanishes and no logarthmic terms appear in the expansion.

\bigskip

In general, the Fefferman-Graham expansion shows that the asymptotic degrees of freedom are encoded in the pair $(g_{(0)},g_{(n)})$, with $g_{(0)}=\gamma$. Indeed, these quantities play the role of asymptotic data in the asymptotic value problem. We define
\begin{definition}\label{defdata}
 A triple $(\Sigma^n,\gamma,g_{(n)})$, where  $(\Sigma^n,\gamma)$ is a smooth Riemannian $n$-manifold and $g_{(n)}$ is a smooth symmetric two-tensor satisfying, according to its dimension, \eqref{constgnodd} or \eqref{constgneven}, is called asymptotic data.
\end{definition}
Establishing a well-posed asymptotic value problem from this formal expansion is, however, a highly non-trivial matter. We refer the reader to \cite{Anderson2005,andersonchrusciel05,kaminski21,kichenassamy03} for treatments in various settings.
More recently, Hintz \cite{hintz23} solved the general problem, providing an alternative proof for a special case of a broader result by Rodnianski and Shlapentokh-Rothman \cite{shlaprothrod}.

For simplicity, we have assumed smoothness of the asymptotic data, although the differentiability requirements may be relaxed (cf. \cite{Fried86initvalue},\cite{hintz23} or \cite{hawkingellis}.) The degree of differentiability of the resulting spacetime metric is determined by the expansions \eqref{expnodd} and \eqref{expneven}.

Then, the characterization by means of asymptotic data can be stated in the following theorem:

\begin{theorem}\label{theoremavp}
Let $(\Sigma^n,\gamma,g_{(n)})$ be asymptotic data and a positive constant $\lambda \in \R$. Then there exists a unique spacetime $( \Sigma^n\times [0,\ell), g)$, for a sufficiently small $\ell \in \mathbb{R}$, with $g$ of the form \eqref{eqggeodgauge}, satisfying \eqref{qEinsteingeodgauge} and whose Fefferman-Graham expansion is determined by the asymptotic data  $(g_{(0)}=\gamma,g_{(n)})$.
\newline
Moreover,  $(\Sigma^n\times (0,\ell), \tilde g = \Om^{-2} g)$ satisfies equation \eqref{Einstein} with the prescribed value of $\lambda$.
\end{theorem}

While in the general case the definition of $g_{(n)}$ depends on the conformal and coordinate gauges, a covariant characterization follows for conformally flat  $(\Sigma^n,\gamma)$ in any dimension. In this case, it can be shown \cite{marspeondata21} that $g_{(n)}$ is given by the electric part of the rescaled Weyl tensor\footnote{For $n=3$ such a characterization holds for general $(\Sigma^n,\gamma)$, as it follows from Friedrich's classical results \cite{friedrich81,friedrich81bis,Fried86initvalue}.}
\begin{equation}\label{geometricgn}
 g_{(n)} = \frac{2\lambda^{-2} }{n(2-n)} \lim_{\Om \to 0} \Omega^{2-n} \mathrm{Weyl}_g(\nabla \Om,\cdot,\nabla \Om,\cdot),
\end{equation}
for all $n \geq 3$ with $n \neq 4$. For $n = 4$ it is the trace-free part of $g_{(n)}$ that satisfies \eqref{geometricgn} and the total coefficient is recovered by adding a trace term constructed from $\gamma$.

Another relevant aspect of the asymptotic value problem is that it must be naturally  conformal and diffeomorphism gauge invariant. This means that
the data is not simply characterized by a manifold $\Sigma^n$ and a pair $( \gamma,g_{(n)})$, but by a pair of {\it conformal and diffeomorphism} equivalence classes $([\gamma],[g_{(n)}])$, where all elements in the equivalence class correspond to diffeomorphic spacetimes.

In general it is hard to determine $[g_{(n)}]$, specially in the $n$ even case, because of its non-zero trace and divergence. This difficulty is removed if $(\Sigma^n,\gamma)$ is locally conformally flat, due to the geometric characterization in \eqref{geometricgn}. In this case, two pairs $(\gamma,g_{(n)})$ and $(\gamma',g'_{(n)})$ belong to the same class $([\gamma],[g_{(n)}])$ if for some diffeomorphism $\varphi$ of $\Sigma$ it holds that
\begin{equation}\label{confdiffeoclass}
 \varphi^\star(\gamma) = \omega^2 \gamma', \qquad \varphi^\star(g_{(n)}) = \omega^{2-n} g'_{(n)},\qquad  \omega \in \mathcal C^\infty(\Sigma^n).
\end{equation}

\bigskip

We can now formulate the asymptotic characterization of de Sitter that we shall require later (see e.g. \cite{marspeondata21} and \cite{skenderis} for the same result with $\lambda <0$):

\begin{corollary}\label{coroldS}
Let $(\Sigma^n,\gamma)$ be a conformally flat Riemannian manifold. The asymptotic data $(\Sigma^n, \gamma, g_{(n)} = 0)$ determines a spacetime $(\Sigma^n \times (0,\ell), \tilde{g})$ locally isometric to de Sitter whose Fefferman-Graham expansion converges to
\begin{equation}\label{FGexpdS}
\tilde{g} = \frac{1}{\Omega^2}\left( -\frac{d\Omega^2}{\lambda} + \gamma + \frac{\Omega^2}{\lambda} P_\gamma + \frac{\Omega^4}{4\lambda^2} P^2_\gamma \right),
\end{equation}
with $P_\gamma$ denoting the Schouten tensor of $\gamma$.
\end{corollary}

\begin{remark}\label{remarkglobal} Corollary~\ref{coroldS} generally determines the spacetime in a neighbourhood of $\mathscr{I}^+$, however, the characterization becomes global when the initial surface is the unit round sphere $(\Sigma^n,\gamma) \equiv (\mathbb{S}^n, \gamma_{\mathbb{S}^n})$. In this case, the Schouten tensor satisfies $P_{\gamma_{\mathbb{S}^n}} = \frac{1}{2} \gamma_{\mathbb{S}^n}$, so that the expansion \eqref{FGexpdS} becomes
\[
\tilde{g} = \frac{1}{\Omega^2}\left( -\frac{d\Omega^2}{\lambda} + \left(1 + \frac{\Omega^2}{4\lambda}\right)^2 \gamma_{\mathbb{S}^n} \right),
\]
which is smooth and convergent for all $\Omega \in (0,\infty)$. Under the change of coordinate $\Omega = 2\lambda^{1/2}\tan(\eta/2)$, this metric is globally isometric to de Sitter within the entire manifold $\mathbb{S}^n \times (0,\pi)$, covering the region from future conformal infinity $\mathscr{I}^+$ to past conformal infinity $\mathscr{I}^-$.
\end{remark}

\subsection{Killing initial data}

The presence of Killing vectors is reflected in the asymptotic data. Following Friedrich's approach to the asymptotic initial value problem in four dimensions, Paetz completely characterized Killing vectors in terms of asymptotic data \cite{KIDPaetz}. In higher dimensions, the general problem remains open. Nevertheless, the analytic case was addressed in \cite{marspeondata21}, where the following result is proved:

\begin{theorem}\label{theoKIDanal}
Let $(\Sigma^n,\gamma,g_{(n)})$ be asymptotic data for a
$\lambda$-vacuum spacetime $(\tilde M,\tilde g)$ with $\lambda>0$, with zero obstruction tensor (i.e. no logarithmic terms in the Fefferman-Graham expansion.) If $(\tilde M,\tilde g)$ admits a Killing vector $X$, then
$\xi:=X|_{\scri}$ is a CKV of $(\Sigma,\gamma)$
satisfying the Killing initial data (KID) equation
\begin{equation}\label{eqKID}
\mathcal L_\xi g_{(n)}
+\frac{n-2}{n}(\operatorname{div}_\gamma\xi)\,g_{(n)}
=0.
\end{equation}
Conversely, if $(\gamma,g_{(n)})$ is analytic and
$\xi$ is a CKV satisfying the above equation,
then there exists a unique Killing vector $X$ on $(\tilde M,\tilde g)$
such that $X|_{\scri}=\xi$.
\end{theorem}

\begin{remark}
 The Theorem is expected to hold beyond the analytic class of data, although up to the author's knowledge, no formal proof has yet been established.
\end{remark}

\begin{definition}
  A set $(\Sigma^n,\gamma,g_{(n)},\xi)$, for asymptotic data $(\Sigma^n,\gamma,g_{(n)})$ and a CVK $\xi$ of $\gamma$ satisfying \eqref{eqKID} is called asymptotic Killing initial data (KID).
\end{definition}

For the analysis in Section \ref{secsufficient}, we need to understand how the conformal and diffeomorphism invariance of the asymptotic problem is reflected in the asymptotic KID. As mentioned in the discussion before equation \eqref{confdiffeoclass}, this is a hard question in general, but it is clear for conformally flat $(\Sigma^n,\gamma)$. Thus, we restrict the following arguments to this case.

Consider the conformal group of a conformally flat $(\Sigma^n,\gamma)$
\begin{equation}
 \mathrm{Conf}(\Sigma^n)=\{\varphi\in\mathrm{Diff}(\Sigma^n)\tq \varphi(\gamma)=\omega^2\gamma,\, \omega\in\mathcal{C}^\infty(\Sigma^n)\}.
\end{equation}
Writing the asymptotic KID as $(\Sigma,\gamma,g_{(n)},\varphi_\star(\xi'))$, where $\xi' = \varphi^{-1}_\star (\xi)$, we obtain the equivalence of data
\begin{equation}\label{equivKID}
 (\Sigma,\gamma,g_{(n)},\varphi_\star(\xi')) \equiv (\Sigma,\varphi^\star(\gamma),\varphi^\star(g_{(n)}),\varphi^\star(\varphi_\star(\xi'))) = (\Sigma,\omega^2\gamma,\omega^{2-n} g'_{(n)},\xi') \equiv (\Sigma,\gamma, g'_{(n)},\xi').
\end{equation}
The first equivalence is diffeomorphism covariance and the last one follows from equivalence \eqref{confdiffeoclass} and conformal invariance of \eqref{eqKID}.

Equivalence \eqref{equivKID} has particularly strong consequences in the de Sitter case
\begin{lemma}\label{lemmaequivKIDdS}
Let $(\Sigma^n,\gamma,g_{(n)}=0,\xi)$ and $(\Sigma^n,\gamma,g_{(n)}=0,\xi')$ be asymptotic KIDs for de Sitter spacetime such that $\xi' = \varphi_\star(\xi)$ for some $\varphi \in \conf{\Sigma^n}$. Then, there exist
spacetimes locally isometric to de Sitter, $(\tilde M^{n+1},\tilde g_{dS})$ and $(\tilde M'^{n+1},\tilde g'_{dS})$, equipped with Killings $X$ and $X'$ respectively, and an isometry $\phi :  (\tilde M^{n+1},\tilde g_{dS}) \to (\tilde M'^{n+1},\tilde g'_{dS})$ such that $\phi_\star(X) = X'$.

\end{lemma}

\begin{remark}
 The result is global provided that $(\Sigma^n,\gamma) \equiv (\mathbb{S}^n,\gamma_{\mathbb{S}^n})$ (cf. Remark \ref{remarkglobal}) and $\xi, \xi'$ are complete CKVs of $(\mathbb{S}^n,\gamma_{\mathbb{S}^n})$.
\end{remark}

Lemma \ref{lemmaequivKIDdS} can also be regarded as a manifestation of the maximal symmetry of de Sitter spacetime. Indeed, bulk isometries of de Sitter are in one-to-one correspondence with conformal transformations of $\scri$. Consequently, all asymptotic KIDs belonging to the same conformal class determine equivalent spacetimes with Killing vectors, up to isometry. This formulation will be useful in Section \ref{secsufficient}, where Killing horizons are constructed directly from asymptotic data.

 \section{Intersection of Killing horizons with $\scri^+$}\label{secnecessary}

 In this section we consider a $(\Lambda>0)$-vacuum spacetime $(\tilde M^{n+1}, \tilde g)$, admitting a conformal extension $(M^{n+1},  g)$ as well as a Killing vector $X$ with a Killing horizon $\mathcal{H}_X$ such that $\closure{\mathcal{H}_X}\cap \scri^+ \neq \emptyset$. We find that this condition provides a local characterization of de Sitter spacetime near the intersection points. This has already been analyzed in the four dimensional case in \cite{kaminskisymmetries}. Here, we extend the result to arbitrary dimensions and include some useful lemmas.

The Killing vectors of spacetimes admitting a conformal extension in turn extend to $\scri^{+}$ as a tangential CKV $\xi = X\mid_{\scri^+}$ of the induced geometry of $\scri^+$ (cf.\cite{marspeondata21}, \cite{KIDPaetz}). Let us denote the set of fixed points of $\xi$ by
\begin{equation}
 \zeroset{\xi} = \{ p \in \scri^+  ~ \slash ~ \xi_ p = 0\}.
\end{equation}
 Then, at the intersection $\closure{\mathcal H_X} \cap \scri^{+}$, $X$ is both null and spacelike, so it is clear that
\begin{equation}
 \closure{\mathcal H_X} \cap \scri^{+} \subseteq \zeroset{\xi}.
\end{equation}
We conclude the following necessary condition for the existence of Killing horizons:
\begin{lemma}\label{lemmazeros1}
 Let $(\tilde M^{n+1},\tilde g)$ be a $(\Lambda>0)$-vacuum spacetime, admitting a conformal extension and a Killing vector field $X$, with $\xi = X\mid_{\scri^+}$. Then $X$ has a Killing horizon such that $ \closure{\mathcal H_X} \cap \scri^{+} \neq \emptyset$ only if the set of fixed points of $\xi$ is non-empty.
\end{lemma}

Therefore,  the zero set $\zeroset{\xi}$ is going to play a central role in the asymptotic study of Killing horizons. Interestingly, in some cases, this set rigidly determines the conformal structure of the manifold as a consequence of the so-called Obata-Ferrand theorem (cf. \cite{obata,obata0,ferrand,schoen}). The theorem states that if the conformal group
\[
\mathrm{Conf}(\Sigma^n)=\{\varphi\in\mathrm{Diff}(\Sigma^n)\tq \varphi^\star(\gamma)=\omega^2\gamma,\, \omega\in\mathcal{C}^\infty(\Sigma^n)\}
\]
of a Riemannian $n$-manifold $(\Sigma^n,\gamma)$ is {\it essential}, namely, it does not act by isometries on any metric in the conformal class $[\gamma]$, then $(\Sigma^n,\gamma)$ must be conformal to the sphere $\mathbb{S}^n$ if it is compact, or to Euclidean space $\mathbb{E}^n$ if it is non-compact.

To connect this result with $\zeroset{\xi}$, it is convenient to formulate an infinitesimal version of the theorem in terms of CKVs. Such a formulation can be found, for instance, in \cite{frances}.
Moreover,
the infinitesimal point of view naturally admits a local version. To state it, we first introduce the following definition:

\begin{definition}\label{defessential}
Let $\xi$ be a CKV of a Riemannian $n$-manifold $(\Sigma^n,\gamma)$. A point  $p\in \Sigma^n$ is an \emph{essential point} of $\xi$ (or, equivalently, $\xi$ is said to be \emph{essential} at $p$) if \emph{there is no metric} $\gamma'$ conformally equivalent to $\gamma$ in a neighbourhood of $p$ with respect to which $\xi$ is Killing. The set of all such points is called the \emph{essential set} of $\xi$.
\end{definition}

The essential set is relevant here because of its relation with the zero set of $\xi$. Indeed, the essential set is necessarily contained in $\zeroset{\xi}$ because, near any point $p$ where $\xi_p\neq 0$, one can always locally solve the transport (ODE) equation for $\omega$
\begin{equation}
 \xi(\omega)   + \frac{\omega}{n}  \dvg_\gamma \xi  = 0 \quad \iff \quad \mathcal{L}_\xi (\omega^2 \gamma) = 0.
\end{equation}
Namely, $p$ is non-essential if $\xi_p \neq 0$.
Furthermore, it is shown in \cite{belgunmoroianuornea} that the essential set of a CKV is a subset of its \emph{isolated zeros}. Combining this with the local results of \cite{frances}, one obtains the following theorem.

\begin{theorem}{\cite{frances,belgunmoroianuornea}\label{theoisolatedzeros}}
Let $\xi$ be a CKV on a Riemannian $n$-manifold $(\Sigma^n, \gamma)$ with zero set $\mathcal{Z}(\xi)$. Then, the essential set of $\xi$ is a subset of $\mathcal{Z}(\xi)$. For any essential zero $p \in \mathcal{Z}(\xi)$, the following hold:
\begin{enumerate}
    \item The point $p$ is isolated in $\mathcal{Z}(\xi)$.
    \item The metric $\gamma$ is conformally flat in a neighbourhood $U$ of $p$.
\end{enumerate}
Moreover, if $\xi$ is complete, then $(\Sigma^n, \gamma)$ is globally conformally flat (i.e., $U = \Sigma^n$).
\end{theorem}

\begin{remark}\label{remarkessentialzeros}
 A classical result due to Kobayashi \cite{kobayashi} shows that the zero set of a Killing vector field has even codimension. If a CKV $\xi$ is non-essential at an isolated zero $p$, then by definition, it is locally Killing with respect to some conformally equivalent metric. Since the codimension of an isolated zero is $n$, this  can only occur when $n$ is even. Thus, in odd dimensions, the essential set of $\xi$ coincides precisely with its set of isolated zeros.
\end{remark}

We can now use Theorem \ref{theoisolatedzeros} to sharpen Lemma \ref{lemmazeros1}. First, we prove the following result:

\begin{proposition}\label{lemmanonessential}
 Let $(\tilde M^{n+1}, \tilde g)$ be a $(\Lambda>0)$-vacuum spacetime admitting a conformal extension and a  Killing vector $X$, with  $\xi =X\mid_{\scri^+}$. If $\xi$ is non-essential in some neighbourhood $U \subseteq \scri^+$, then $X$ is a spacelike Killing in a neighbourhood $V\subseteq M$ of $U$, therefore, it has no horizons in $V$.
 \end{proposition}

 \begin{proof}
  Consider a conformal extension in the geodesic conformal gauge \eqref{geodesicgauge}.   Denoting $\phi = \dvg_g X/((n+1)$, the conformal Killing equation for $X$ reads
   \begin{equation}
    \mathcal{L}_X g = 2 \phi g.
   \end{equation}
   Contracting this equation twice with $\nabla^a\Om $, the RHS is simply $-2 \lambda \phi$ and the LHS yields
   \begin{align}
    \nabla^a \Om \nabla^b  \Om \mathcal L_X g_{ab} = - g_{ab}\mathcal{L}_X (\nabla^a \Om \nabla^b\Om) = -2 \nabla_a \Om (X^b \nabla_b \nabla^a \Om - \nabla^b \Om \nabla_b X^a ) = 2 \nabla^b \Om \nabla_b (\nabla_a \Om X^a),   \end{align}
    where in the last equality we have used the fact that $\nabla^a \Om$ is a geodesic field, thus $\nabla^a \Om \nabla_a \nabla_b \Om = 0$. In coordinates $\{\Omega,x^A\}$ adapted to the level sets of $\Omega$, $\nabla^a \Om = -\lambda \partial_\Om$ and denoting $\nabla_a \Om X^a = X(\Om)$, we have
    \begin{align}
    \nabla^a \Om \nabla^b  \Om \mathcal L_X g_{ab} = - 2 \lambda \partial_\Om X(\Om). \end{align}
Now observe that since $X$ is a Killing vector of the metric $\tilde g$, the following identity holds
\begin{equation}\label{XOmdivX}
 \mathcal{L}_X(\tilde g) = \mathcal{L}_X(\Om^{-2} g) = -\frac{2 X(\Om)}{\Om^3} g + \frac{1}{\Om^2} \mathcal{L}_X(g) = 0,\quad  \iff \quad X(\Om) = \frac{\Om}{2} \mathcal{L}_X(g) = \Om \phi.
\end{equation}
Therefore the LHS of the $\Om-\Om$ component of the conformal Killing equation is
\begin{align}
  - 2 \lambda \partial_\Om X(\Om)= - 2 \lambda \phi - 2 \lambda \Om \partial_\Om \phi,    \end{align}
and equating to the RHS it reduces to
\begin{equation}\label{pdephi}
 2 \lambda \Om \partial_\Om \phi = 0 \iff  \partial_\Om \phi = 0.
\end{equation}
The initial data for this equation is determined by (see the proof in Section \ref{secgrav},  equation \eqref{divXdivxi})
\begin{equation}\label{dataphi}
 \phi\mid_\scri =  \frac{1}{n+1} \dvg_g X\mid_\scri = \frac{1}{n} \dvg_\gamma \xi.
\end{equation}
Now recall from the discussion in Section \ref{secdata} that we still have freedom to choose a geodesic conformal extension for each representative in $[\gamma]$, the conformal structure at $\scri^+$. Since $\xi$ is non-essential in $U$, we choose a conformal extension inducing a metric $\gamma \in [\gamma]$ at $\scri^+$ for which $\xi$ is a Killing in $U$. Hence by \eqref{dataphi} we have $\phi\mid_U = 0$ and by \eqref{pdephi} immediately follows $\phi = 0$ in a neighbourhood $V\subset M^{n+1}$ of $U$. This, in turn, implies  by \eqref{XOmdivX} that $X(\Om) = 0$ in $V$ so that $X$ is Killing of $g$ tangential to the $\{ \Om = const.\}$ submanifolds. The Proposition now follows from the fact that the $\{ \Om = const.\}$ submanifolds are spacelike near $\scri^+$. Thus, reducing $V$ if necessary,  $X$ is spacelike in $V$ and therefore it cannot have Killing horizons.
 \end{proof}

The combination of Proposition  \ref{lemmanonessential} and Theorem \ref{theoisolatedzeros} has strong consequences:
\begin{proposition}\label{propisolatedzeros}
 Let $(\tilde M^{n+1},\tilde g)$ be a $(\Lambda>0)$-vacuum spacetime, admitting a conformal extension and a Killing vector field $X$, with $\xi = X\mid_{\scri^+}$. Then $X$ has a Killing horizon $\mathcal{H}_X$ such that $\closure{\mathcal{H}_X} \cap \scri^+ \neq \emptyset $ only if $\xi$ is essential at $\closure{\mathcal{H}_X} \cap \scri^+$. Namely,  $\closure{\mathcal{H}_X} \cap \scri^+$ is a subset of the isolated zeros of $\xi$ and  $\gamma$ is conformally flat in a neighbourhood $U$ of each one of these points.

 Moreover, if $\xi$ is complete, then $U = \scri^+$.
 \end{proposition}
 \begin{proof}
 By Proposition \ref{lemmanonessential}, it follows immediately that only if $\xi$ is essential, $X$ can have a Killing horizon $\mathcal{H}_X$ such that $\closure{\mathcal{H}_X} \cap \scri^+ \neq \emptyset $. The intersection point $p$ must be in the zero set of $\xi$ and thus, must be an isolated zero by Theorem \ref{theoisolatedzeros}. Also by this Theorem, $\gamma$ is  conformally flat in a neighbourhood $U$ of $p$  and $U = \scri^+$ if $\xi$ is complete.
\end{proof}

We now make use of Theorem \ref{theoKIDanal}. Under the hypotheses of Proposition \ref{propisolatedzeros}, let $p \in \overline{\mathcal{H}_X}\cap \scri^+$ and $U$ a conformally flat neighbourhood of $p$. Using the conformal Killing equation we can write
\begin{equation}\label{divxiliexi}
 \mathcal{L}_\xi(\gamma)(\xi,\xi) =2\mathcal{L}_\xi(|\xi|^2_\gamma) = \frac{2|\xi|_\gamma^2}{n} \dvg_\gamma\, \xi\quad \iff \quad \dvg_\gamma\, \xi = 2 n \frac{\mathcal{L}_\xi(|\xi|_\gamma)}{|\xi|_\gamma}.
\end{equation}
By Theorem \ref{theoKIDanal}, the asymptotic data satisfies the KID equation \eqref{eqKID}, which in combination with \eqref{divxiliexi}, can be rewriten as
\begin{equation}\label{transportKID}
\mathcal L_\xi g_{(n)}
+2(n-2)\frac{\mathcal{L}_\xi(|\xi|_\gamma)}{|\xi|_\gamma}\,g_{(n)}
=0,\quad \iff \quad \mathcal L_\xi\left( |\xi|_\gamma^{2(n-2)}g_{(n)} \right)=0.
\end{equation}
Thus $|\xi|_\gamma^{2(n-2)}g_{(n)}$ is constant along the integral lines of $\xi$.
By the classification of essential CKVs of $\S^n$ in \cite{obata0}, if follows that in the conformally flat neighbourhood $U$, every integral curve of $\xi$ converges to (or from) the essential isolated zero $p$. Thus, shrinking $U$ if necessary, all of its points are connected with $p$ by an integral line contained in $U$. Hence, the integration of \eqref{transportKID} yields that $|\xi|_\gamma^{2(n-2)}g_{(n)}$ is constant in $U$, namely
\begin{equation}
 \left(|\xi|_\gamma^{2(n-2)}g_{(n)}\right)(p) = 0, \quad \Longrightarrow \quad |\xi|_\gamma^{2(n-2)}g_{(n)}\mid_U = 0, \quad \Longrightarrow \quad g_{(n)}\mid_U = 0.
\end{equation}
Therefore, as a consequence of having a Killing horizon intersecting $\scri^+$, we find that the local asymptotic data is $(U,\gamma,g_{(n)} = 0)$. By Corollary \ref{coroldS}, this implies that the spacetime is locally isometric to de Sitter. Moreover, by Theorem \ref{theoisolatedzeros}, if $\xi$ is a complete CKV with essential zeros then $U = \scri^+$ is conformally flat and the result applies to the whole domain of dependence of $\scri^+$.

In summary, we have proven:

\begin{theorem}\label{theoremnecessary}
  Let $(\tilde M^{n+1},\tilde g)$ be a $(\Lambda>0)$-vacuum admitting a conformal extension and a Killing vector field $X$, with $\xi = X\mid_{\scri^+}$. Assume that $X$ has a Killing horizon $\mathcal{H}_X$ such that $\closure{\mathcal{H}_X}\cap \scri^+ \neq \emptyset$. Then $\xi$ has essential isolated zeros at $\closure{\mathcal{H}_X}\cap \scri^+$ and for each $p \in \closure{\mathcal{H}_X}\cap \scri^+$ there exist conformally flat neighbourhoods $U_p \subseteq \scri^+$. Moreover $(\tilde M^{n+1},\tilde g)$ is isometric to de Sitter in the domain of dependence of each $U_p$.

  In addition, if $\xi$ is complete,  $U_p = \scri^+$. In particular, if  $ \scri^+=\S^n$, $(\tilde M^{n+1},\tilde g)$ is globally isometric to de Sitter spacetime.
\end{theorem}

\begin{remark}\label{remarksuff}
 The reader may note that a converse to Theorem \ref{theoremnecessary} is  straightforward. Namely, that all the horizons of de Sitter spacetime intersect $\scri^+$ (e.g. \cite{gallowaydS1,gallowaydS2}) and moreover, by Proposition \ref{propisolatedzeros},  the intersection must happen at the essential isolated zeros of $\xi$. Therefore, a Killing horizon with non-empty intersection with $\scri^+$ (locally) characterizes de Sitter spacetime.

 However, a more refined version of this statement requires determining which types of CKVs of $\scri^+$ with essential zeros extend to Killing vectors admitting a horizon which ones do not. We address this characterization in Section \ref{secsufficient}.
\end{remark}

\section{Asymptotic characterization of horizons in de Sitter}\label{secsufficient}

In the previous section we showed (Theorem \ref{theoremnecessary}) that
a spacetime admitting a Killing vector $X$ with a  Killing horizon whose closure intersects $\scri^+$ must be isometric to de Sitter spacetime near the intersection points, which coincide with the essential isolated zeros of  $\xi = X \mid_{\scri^+}$. Thus an arbitrary CKV $\xi$ of $\scri^+$ in de Sitter extends to a Killing vector $X$ of the bulk spacetime only if it admits essential isolated zeros. This is a necessary however not a sufficient condition. In this section, we determine the sufficient condition that $\xi$ must satisfy in order for its extension $X$ to admit a Killing horizon, thereby asymptotically characterizing and classifying all the Killing horizons of de Sitter spacetime.

For the rest of this section, we assume that we are in de Sitter spacetime. The strategy is to consider an arbitrary CKV $\xi$ of $\scri^+$ and, by direct construction, we obtain the unique Killing vector $X$ extending from $\xi$ and evaluate the conditions under which it has a Killing horizon.

In order to obtain a global picture, one ideally works in a conformal extension such that $\scri^+$ is identified with $\mathbb S^n$. For computational purposes, however, it is simpler to use a conformal extension of de Sitter spacetime for which $\scri^+$ is isometric to the Euclidean space $\mathbb E^n$, endowed with Cartesian coordinates $\{y^A\}$. Global statements can nevertheless be recovered  mapping back to $\mathbb S^n$ by the one-point compactification map $\mathbb S^n \equiv \mathbb{E}^n \cup \{ \infty \}$, which we consider throughout this section.

\bigskip

In Cartesian coordinates, the CKVs of $\mathbb{E}^n$ are (see Appendix \ref{appcanonical}):
\begin{equation}\label{CKVFgeneralsec}
\xi=
\lr{
\b^A+\nu y^A+(\a_B y^B)y^A
-\frac12(y_B y^B)\a^A
-{\omega^A}_B y^B
}\partial_{y^A}.
\end{equation}
By Lemma \ref{lemmaequivKIDdS},  asymptotic KIDs of de Sitter $(\mathbb{E}^n,g_{(n)}=0,\varphi_\star(\xi))$ are equivalent  for all $\varphi \in \conf{\mathbb{E}^n}$. Therefore, we can equivalently select data $(\mathbb{E}^n,g_{(n)}=0,\xi)$, with $\xi$ being any element $\xi \in [\xi]$, where
\begin{equation}\label{confclass}
 [\xi] = \{\varphi_\star(\xi) \in \varphi \in \conf{\mathbb{E}^n} \}
\end{equation}
is the orbit of $\xi$ by conformal diffeomorphisms. As discussed in Appendix \ref{appcanonical}, there is a canonical representative $\xi_{can} \in [\xi]$ for each one of these classes, constructed from a number of conformal invariants, which allows us to tackle the problem at once.

\bigskip

We start by discussing the $n$-even case and use the arguments and, based on this one, we later complete the discussion to the $n$-odd case. Let us denote the Cartesian coordinates $\{ y^A\} = \{z_1,z_2,\{x_i\}_{i=1}^{2p} \}$, with
\begin{equation}
 p = [(n+1)/2]-1 = n/2-1.
\end{equation}
Then the canonical representative $\xi_{can}$ of the orbit $ [\xi]$ is
\begin{equation}\label{xicanevensec}
 \xi_{can} = \tilde \xi + \sum_{i=1}^p \mu_i \eta_i,
\end{equation}
where
 \begin{align}
  \tilde \xi &= \frac{1}{2}\Big(\aa + z_1^2 - z_2^2 - \sum\limits_{i=1}^{2p}x_i^2 \Big) \partial_{z_1} + \lr{\frac{\bb}{2} + z_1 z_2} \partial_{z_2} + z_1 \sum\limits_{i=1}^{2p} x_i \partial_{x_i},\label{CKVFxi}\\
  \eta_i & = x_{2i-1} \partial_{x_{2i}} - x_{2i} \partial_{x_{2i-1}},\label{CAKVF}
 \end{align}
are orthogonal CKVs. The parameters $\{\aa,\bb,\mu_i\}_{i=1}^p$ determining the canonical form, with
\begin{equation}
 \aa \in \mathbb{R},\qquad \bb,\mu_i \in \mathbb{R}^+ ~(= \{a \in \mathbb{R} \tq a \geq0\} \},
\end{equation}
are conformal invariants. These are, however, not easily extracted from an arbitrary CKV $\xi \in [\xi_{can}]$ of the form \eqref{CKVFgeneralsec}.  For completeness, Appendix \ref{appcanonical} describes how to compute these parameters and provides relevant references.

\bigskip

By a classical result \cite{blairzeros}, the zero set of a CKV of $\S^n$ consists of either isolated points or of submanifolds of even codimension. Its dimension is a conformal invariant and is therefore shared by all elements of $[\xi]$. In the flat representative $\mathbb{E}^n$, the zeros of $\xi$ can be computed directly from \eqref{CKVFgeneralsec}, except possibly for isolated zeros located at infinity, i.e., those corresponding, under a suitable inverse stereographic projection, to a pole of $\mathbb{S}^n$. The following result, proved in \cite{KdSlike}, characterizes precisely when this occurs.

\begin{lemma}\label{lemmazerosinfty}
Let $\xi$ be a CKV of $\mathbb{E}^n$ of the form \eqref{CKVFgeneralsec}. Then $\xi$ has an isolated zero at infinity if and only if  $\a^A = 0$.
\end{lemma}

 We are next going to reconstruct the horizons from $\scri^+$ by extending $\xi_{can}$ uniquely to a Killing vector $X$ of the spacetime. As argued at the beginning of this section, we know that  the Killing horizons in de Sitter intersect $\scri^+$ at the essential isolated zeros of $\xi_{can}$. Thus,  Lemma \ref{lemmazerosinfty} is needed to ensure that we are not missing any possible horizon intersection as a consequence of the conformal representative we are using.  From \eqref{CKVFxi}, it is clear that the canonical form \eqref{xicanevensec} has $\a^A = \delta^A{}_{z_1}$, so in no case any of the zeros of $\xi_{can}$ appear at infinity.

\bigskip

In the conformal extension and coordinates we have chosen, $\scri^+ = \mathbb{E}^n$ with metric
\begin{equation}
 \gamma = \delta_{AB} d y^A \otimes d y^B,
\end{equation}
so by Corollary \ref{coroldS}, we can write de Sitter as a metric explictitly conformal to Minkowski
\begin{equation}\label{dSconftomink}
 \tilde g_{dS} = \frac{1}{\Om^2} g_M,\qquad\qquad g_M = \lambda^{-1}\left( -d \Om^2 + \gamma \right).
\end{equation}

In \cite{kaminskisymmetries} the four-dimensional Killing equation is integrated
\footnote{We note that the expression in \cite{kaminskisymmetries} is missing the $1/n$ factor (for $n=3$) in the coefficient of $\Omega^2$.}
 in Cartesian coordinates $\{\Om,y^A \}$. The same arguments are easily extended to arbitrary dimensions and in our conformal picture it gives the following result:

\begin{proposition}
Consider the local representation of de Sitter metric into the form \eqref{dSconftomink} in Cartesian coordinates $\{\Om,y^A \}$ and let $\xi$ be a  CKV of $\scri$. Then the unique Killing vector $X$ extending $\xi$  off $\scri^+$ is
\begin{equation}\label{killext}
 X = \frac{\Om}{n} (\dvg_\gamma  \xi) \partial_\Om + \frac{\Om^2}{2 n} \gamma^{AB} \partial_{y^B} \left( \dvg_{\gamma} \xi \right) \partial_{y^A} + \xi^A \partial_{y^A}.
\end{equation}

\end{proposition}
\begin{proof}
Trivially, $X$ is an extension of $\xi$. The fact that $X$ is Killing follows by direct computation.  We denote the second term in \eqref{killext} by $\xi_0$, so that
\begin{equation}
 X = X(\Om) \partial_\Om + \xi_0 + \xi,
\end{equation}
and then compute
\begin{equation}\label{killeq}
 \mathcal{L}_X \tilde g_{dS}= -\frac{2}{\Om^3} X(\Om) g_M + \frac{1}{\Om^2} \mathcal{L}_X g_M =  -\frac{2}{n \Om^2} (\dvg_\gamma \xi) g_M + \frac{1}{\Om^2} \mathcal{L}_X g_M,
\end{equation}
where in the second equality we have used that, by \eqref{killext}, $X(\Om) = (\Om/n) \dvg_\gamma \xi$.
Writing
\begin{equation}
 \mathcal{L}_X g_M = \mathcal{L}_{X(\Om)\partial_\Om} g_M + \mathcal{L}_{\xi_0} g_M + \mathcal{L}_\xi g_M,
\end{equation}
and recalling \eqref{dSconftomink}, we have that the first term is given by
\begin{align}
 \lambda \mathcal{L}_{X(\Om)\partial_\Om} g_M & = - d(X(\Om)) \otimes d \Om - d \Om \otimes  d(X(\Om))  \\
 & = -\Big( \frac{\dvg_\gamma \xi}{n} d\Om + \frac{\Om}{n} \partial_{y^A}(\dvg_\gamma \xi) d y^A \Big) \otimes d \Om - d \Om \otimes \Big( \frac{\dvg_\gamma \xi}{n} d\Om + \frac{\Om}{n} \partial_{y^A}(\dvg_\gamma \xi) d y^A \Big),\label{1stterm}
 \end{align}
the second one, taking into account that from \eqref{CKVFgeneralsec} it follows $\partial_{y^D}\partial_{y^C} (\dvg_\gamma \xi) = 0$, yields
\begin{align}
 \lambda \mathcal{L}_{\xi_0} g_M & = \mathcal{L}_{\xi_0} \gamma = \delta_{AB} (d \xi_0(y^A) \otimes d y^B + d y^A \otimes d \xi_0(y^B)) \\ & = \delta_{AB} \Big( \big(\frac{\Om}{n}\delta^{AC}\partial_{y^C} (\dvg_\gamma \xi) d \Om \big) \otimes d y^B  + d y^A \otimes  \big(\frac{\Om}{n}\delta^{BC}\partial_{y^C} (\dvg_\gamma \xi) d \Om \big)\Big),\label{2ndterm}
\end{align}
and the third one
\begin{align}
 \lambda \mathcal{L}_\xi g_M =  \mathcal{L}_\xi \gamma = \frac{2}{n} (\dvg_\gamma \xi)\gamma \label{3rdterm}.
\end{align}
Putting \eqref{1stterm},\eqref{2ndterm} and \eqref{3rdterm} together we obtain
\begin{equation}
 \mathcal{L}_X g_M = \mathcal{L}_{X(\Om)\partial_\Om} g_M + \mathcal{L}_{\xi_0} g_M + \mathcal{L}_\xi g_M = \frac{2}{n}(\dvg_\gamma \xi) g_M
\end{equation}
which back into \eqref{killeq} yields that $X$ is a Killing of $\tilde g_{dS}$. Since de Sitter is analytic, uniqueness follows from Theorem \ref{theoKIDanal}.
\end{proof}

With this result at hand, we can calculate the Killing $X_{can}$  which extends $\xi_{can}$. We first compute
\begin{equation}\label{dvgxican}
 \dvg_\gamma  \xi = \partial_{y^A} \xi_{can}^A = (2 p +1) z_1 = n z_1,
\end{equation}
and therefore, by \eqref{killext},
\begin{align}
X_{can}  & = \Om z_1 \partial_\Om + \frac{\Om^2}{2}\partial_{z_1} + \xi_{can} =
\frac{1}{2} \Big(\aa + \Om^2 + z_1^2 -z_2^2  - \sum\limits_{i=1}^{2p}x_i^2 \Big) \partial_{z_1} + \frac{\bb}{2} \partial_{z_2}
\\ & + z_1 \big(\Om  \partial_\Om + z_2 \partial_{z_2}  + \sum\limits_{i=1}^{2p}x_i \partial_{x_i} \big)   + \sum_{i=1}^p \mu_i ( x_{2i-1} \partial_{x_{2i}} - x_{2i} \partial_{x_{2i-1}}).\label{Xcan}
\end{align}
To evaluate the conditions under which  $X_{can}$ admits a Killing horizon, we define the function
\begin{equation}
 F := |X_{can}|_{g_M}^2.
\end{equation}
The Killing horizons are given by the zero level sets of $F$, wherever $dF$ is null and non-zero.
Note that since the causal character is a conformal invariant, it does not matter whether we define $F$ using de Sitter or Minkowski metric $g_M$. For simplicity, we use $g_M$ and we also denote
\begin{equation}
 \rho_i^2 := x_{2i-1}^2 + x_{2i}^2 ,\qquad \rho^2 :=\sum_{i=1}^{2p}x_i^2 = \sum_{i=1}^p \rho_i^2,
\end{equation}
where $\rho,\rho_i$ will denote the respective positive root. From the first equality in \eqref{Xcan}
\begin{align}\label{eqF1}
 F = -\Om^2 z_1^2 +\frac{\Om^4}{4}   + |\xi|^2+ \Om^2 g_M(\partial_{z_1},\xi_{can}) =  \frac{\Om^4}{4}+\frac{\Om^2}{2}\Big(\aa - z_1^2 - z_2^2 - \rho^2\Big)   + |\xi_{can}|^2
\end{align}
it is easy to compute the roots of $F$
\begin{equation}
 4 F = 0 \iff \Om^2 = - (\aa - z_1^2 -z_2^2- \rho^2) \pm \sqrt{(\aa - z_1^2 -z_2^2- \rho^2)^2 - 4 |\xi_{can}|^2}.
\end{equation}
Taking into account
\begin{align}
 |\xi_{can}|^2 & = \frac{1}{4}\Big(\aa + z_1^2 - z_2^2 - \rho^2 \Big)^2 + \big(\frac{\bb}{2} + z_1 z_2\big)^2  + z_1^2 \rho^2 + \sum\limits_{i=1}^{p} \mu_i^2 \rho_i^2.
\end{align}
a direct computation shows
\begin{equation}
  F = 0 \iff \Om^2 = (z_1^2 + z_2^2 + \rho^2 - \cc) + 2\beta_\pm
\end{equation}
where
\begin{equation}
  \beta_\pm := \pm\Big(-\cc z_1^2 - \sum_i^p \mu_i^2 \rho_i^2 - \left(\frac{\tau}{2}\right)^2 - \tau z_1 z_2\Big)^{1/2}.
\end{equation}
This allows us to factor $F$ as
\begin{equation}
  F = |X|^2 = (\Om^2 - \alpha_+) (\Om^2 - \alpha_-),\qquad
\alpha_\pm := (z_1^2 + z_2^2 + \rho^2 - \cc) + 2\beta_\pm.
\end{equation}
Observe that, given that by definition $\bb \geq 0$, only $\aa \leq 0$ cases admit real roots. This is the first restriction on the paremeters $\{ \aa, \bb, \mu_i \}$ and thus we write $\aa = - |\aa|$ from now on.

\bigskip

The causal character of the zero level sets of $F$  follows by simply evaluating $|dF|_{g_M}^2$ at $\Om^2 = \alpha_\pm$. This is a straightforward but tedious calculation, whose details are included in Appendix \ref{appdF}. The result can be reduced to
\begin{align}
\frac{1}{64}|dF|_{g_M}^2 \mid_{\Om^2 = \alpha_\pm} =   \frac{\tau^2}{4}(z_1^2 +z_2^2 + |\aa| +2 \beta_\pm)  +  \sum_{i=1}^{p}\mu_i^2 \rho_i^2(\mu_i^ 2 + |\cc|).
\end{align}
For $\Om^2 = \alpha_+$ we have
\begin{align}
\frac{1}{64}|dF|_{g_M}^2 \mid_{\Om^2 = \alpha_+} =   \frac{\tau^2}{4}(z_1^2 +z_2^2 + |\aa |+2 \beta_+)  +  \sum_{i=1}^{p}\mu_i^2 \rho_i^2(\mu_i^ 2 + |\cc|) \geq 0
\end{align}
because it is a sum of positive terms. Therefore, the condition for the existence of a \emph{null hypersurface} with normal vector $d F$ is
\begin{align}
\frac{1}{64}|dF|_{g_M}^2 \mid_{\Om^2 = \alpha_+} =   0  \qquad \iff \qquad \bb = \mu_1 = \cdots = \mu_p = 0.
\end{align}
Note that some of the rotation parameters $\mu_i$ may not vanish and still $|dF|_{g_M}^2 = 0$ with $d F \neq 0$  if the corresponding $\rho_i = 0$. This is, however, a lower dimensional submanifold and not a Killing horizon, which is a null hypersurface (see subsection \ref{secdefHx}), so we do not consider these cases further.

On the other hand, for $\Om^2 = \alpha_-$,
\begin{align}
\frac{1}{64}|dF|_{g_M}^2 \mid_{\Om^2 = \alpha_-} =   \frac{\tau^2}{4}(z_1^2 +z_2^2 + |\aa| +2 \beta_-)  +  \sum_{i=1}^{p}\mu_i^2 \rho_i^2(\mu_i^ 2 + |\cc|).
\end{align}
Taking into account that
\begin{align}
  \beta_- & = -\Big(|\cc| z_1^2 - \sum_{i= 1}^p \mu_i^2 \rho_i^2 - \left(\frac{\tau}{2}\right)^2 - \tau z_1 z_2\Big)^{1/2}  \\
  & = -\Big(|\cc| z_1^2 +z_1^2 z_2^2- \sum_{i=1}^p\mu_i^2 \rho_i^2 - \left(\frac{\tau}{2} + z_1 z_2\right)^2 \Big)^{1/2}
  \geq -\big(|\cc| z_1^2 + z_1^2 z_2^2\big)^{1/2}
\end{align}
we find
\begin{align}
\frac{1}{64}|dF|_{g_M}^2 \mid_{\Om^2 = \alpha_-} & \geq   \frac{\tau^2}{4}\big(z_1^2 +z_2^2 + |\aa|-2\big(|\cc| z_1^2 + z_1^2 z_2^2\big)^{1/2}\big)  +  \sum_{i=1}^{p}\mu_i^2 \rho_i^2(\mu_i^ 2 +|\cc|)\\
& = \frac{\tau^2}{4}\Big(\big(|z_1|-(|\aa|+z_2^2)^{1/2}\big)^2 +z_2^2 \Big)  +  \sum_{i=1}^{p}\mu_i^2 \rho_i^2(\mu_i^ 2 + |\cc|) \geq 0.
\end{align}
Hence, the condition for a null hypersurface with $d F$ as normal vector is
\begin{align}
\frac{1}{64}|dF|_{g_M}^2 \mid_{\Om^2 = \alpha_-} =   0\qquad \iff \qquad \bb = \mu_1 = \cdots = \mu_p = 0.
\end{align}
Note that in this case a lower dimensional null submanifold may also occur if $\mu_i\neq 0$ at $\rho_i = 0$ and if $\bb \neq 0$ at $z_2 = 0$ and $|z_1| = |\cc|^{1/2}$, which we do not consider as we shall focus on hypersurfaces.

\bigskip

In summary, we have proven that $X_{can}$ admits a Killing horizon if only if the parameters $\{\aa,\bb,\mu_i\}$ are
\begin{equation}\label{paramswhorizoneven}
\aa \leq 0,\qquad \bb = \mu_1 = \cdots = \mu_p = 0.
\end{equation}
Therefore
\begin{align}\label{xicanhorizonsneven}
\xi_{can}  & = \frac{1}{2} \Big(\aa  + z_1^2 -z_2^2  - \sum\limits_{i=1}^{2p}x_i^2 \Big) \partial_{z_1}  + z_1 \big( z_2 \partial_{z_2}  + \sum\limits_{i=1}^{2p}x_i \partial_{x_i} \big)
\end{align}
and
\begin{align}
X_{can}  & = \frac{1}{2} \Big(\aa + \Om^2 + z_1^2 -z_2^2  - \sum\limits_{i=1}^{2p}x_i^2 \Big) \partial_{z_1}  + z_1 \big(\Om  \partial_\Om + z_2 \partial_{z_2}  + \sum\limits_{i=1}^{2p}x_i \partial_{x_i} \big)   .
\end{align}

As an interesting check of consistence with Theorem \ref{theoremnecessary}, it follows from Lemma \ref{zerosetxican} that $\xi_{can}$ is a particular case of a CKV with isolated zeros. More geometrically, $\xi_{can}$ is conformally equivalent to a homothety if $\aa<0$ and to a translation if $\aa=0$ (cf.~Remark~\ref{remarkrotation}), both cases having essential isolated zeros (see the classification in \cite{obata0}), in accordance with the results of  Section~\ref{secnecessary}.

\bigskip

From the analysis above, the locus of the horizons reduces to
\begin{equation}\label{horizonsreduced}
 \Om^2 = (z_1 \pm |\aa|^{1/2})^2 + z_2^2 +\sum\limits_{i=1}^{2p}x_i^2.
\end{equation}
Note that the negative and positive values of $\Om$ belong to different sides (future and past) of $\scri^+$, so we focus only on $\Om >0$ values.

For $|\sigma|\neq 0$ there are two branches
\begin{equation}
 \mathcal{H}_+  = \{-\Om^2 + (z_1 + |\aa|^{1/2})^2 + z_2^2 +\sum_i x_i^2 = 0 \},\qquad \mathcal{H}_-  = \{ -\Om^2 + (z_1 - |\aa|^{1/2})^2 + z_2^2 +\sum_i x_i^2 = 0\},
\end{equation}
intersecting at
\begin{equation}
S = \mathcal{H}_+ \cap \mathcal{H}_- =\{-\Om^2 + |\aa| + z_2^2 +\sum_i x_i^2 = 0,\quad z_1 = 0 \}.
\end{equation}
This case corresponds to a bifurcate Killing horizon with bifurcation surface $S$. In the conformal gauge that we are using, only the local topology of $S$ is apparent, namely, that of a hyperboloid. We can recover its global topology by looking at de Sitter spacetime as the $(n+1)$-dimensional hyperboloid embedded in $(n+2)$-dimensional Minkowski spacetime, i.e.,
\begin{equation}
 \mathbb{H}^{n+1} = \{(X_0,X_1,\cdots,X_{n+1}) \in \mathbb{R}^{n+2} \tq -X_0^2 + X_1^2 + \cdots + X_{n+1}^2 = \lambda^{-1} \}.
\end{equation}
A direct calculation shows that the coordinates $\{\Om,y^A \} = \{\Om,z_1,z_2,x_i \} $ correspond to the local parametrization
\begin{equation}
 X_0 = \frac{\lambda^{-1} - \Om^2 + r^2}{2\Om}  ,\qquad X_1 = \frac{\lambda^{-1} + \Om^2 - r^2}{2\Om} ,\qquad  X_{A+1} = \lambda^{-1/2}\frac{y^A}{\Om},
\end{equation}
where $r^2 = \sum_A (y^A)^2 $. Then the bifurcation surface is the embedded $(n-1)$-spheroid
\begin{equation}
 S = \{(X_0,X_1,\cdots,X_{n+1}) \in \mathbb{R}^{n+2} \tq X_0 = \kappa X_1,\quad X_2 = 0,\quad (1-a^2) X_1^2 + X_3^2 + \cdots + X_{n+1}^2 = \lambda^{-1}\},
\end{equation}
with $a:= \frac{\lambda^{-1} - |\aa|}{\lambda^{-1} + |\aa|}$.

As  $|\aa|$ takes smaller values, it is clear that by \eqref{horizonsreduced} the two horizons $\mathcal{H}_+$ and $\mathcal{H}_-$ approach and then merge into a single hypersurface when $|\aa| = 0$,
\begin{equation}
 \mathcal{H}_0  = \{-\Om^2 + z_1^2 + z_2^2 +\sum_i x_i^2 = 0 \}.
 \end{equation}

 In de Sitter spacetime, bifurcate and non-bifurcate Killing horizons correspond, respectively, to non-degenerate (non-zero surface gravity, $\kappa \neq 0$) and degenerate (zero surface gravity, $\kappa = 0$) cases. This can be confirmed by a direct calculation of the surface gravity. We will verify this fact in Section~\ref{secgrav}, where we provide a formula that allows for a straightforward computation of the surface gravity at $\scri^+$ in terms of the conformal invariant $\aa$, namely $\kappa^2 = |\sigma|$.

\begin{remark}\label{remarkgravity}
We note that de Sitter Killing horizons are multiple Killing horizons of maximal dimension in the sense of \cite{marsmultiplekill}, i.e. they admit $n+1$ tangential null Killing vectors, only one having non-zero surface gravity. Although our analysis in this work focuses on a single generator, extending our construction at $\scri^+$ to recover the multiplicity properties asymptotically is an interesting future application.
\end{remark}

 \bigskip

 We now turn into the $n$ odd dimensional case. Let us denote the Cartesian coordinates of $\mathbb{E}^n$ by $\{ y^A\} = \{z_1,\{x_i \}_{i=1}^{2p}\}$, with
\begin{equation}
 p = [(n+1)/2]-1 = (n-1)/2.
\end{equation}
Then the canonical representative $\xi_{can}$ of the orbit $ [\xi]$ is in this case
\begin{equation}\label{xicanodd}
 \xi_{can} = \tilde \xi + \sum_{i=1}^{p} \mu_i \eta_i,\\
\end{equation}
with
 \begin{align}
  \widetilde \xi &= \Big(\frac{\aa}{2} + \frac{1}{2}\big(z_1^2  - \sum\limits_{i=1}^{2p}x_i^2 \big)\Big) \partial_{z_1}  + z_1 \sum\limits_{i=1}^{2p} x_i \partial_{x_i}, \label{CKVFxioddsec}\\
  \eta_i & = x_{2i-1} \partial_{y_{2i}} - x_{2i} \partial_{x_{2i-1}},\label{CAKVFoddsec}
 \end{align}
 orthogonal CKVs and the parameters $\{\sigma,\mu_i\}_{i=1}^p$ are also conformal invariants (see details in Appendix \ref{appcanonical}.)

The analysis of the $n$ odd case can be recovered from that of the even-dimensional one by means of the following  embedding argument. Consider the  isometric embedding of the odd-dimensional Euclidean space $\mathbb{E}^n$ into the even-dimensional space $\mathbb{E}^{n+1}$ given by
\begin{equation}
 \iota: \mathbb{E}^n \hookrightarrow \mathbb{E}^{n+1},\qquad \iota(\mathbb{E}^n) = \{z_2 = 0 \}.
\end{equation}
Comparing the canonical forms \eqref{xicanodd} and \eqref{xicanevensec}, it is clear that the pushforward $\iota_\star(\xi_{can})$ of the odd-dimensional canonical CKV coincides with an even-dimensional canonical CKV evaluated at $z_2=0$ with $\tau = 0$.
Similarly, the even-dimensional bulk Killing  field \eqref{Xcan} with $\tau = 0$ is tangent to the bulk hypersurface $\{z_2 = 0\}$ and it is a Killing of the induced metric. This means that in the $n$ odd case, by Theorem \ref{theoKIDanal}, $\xi_{can}$ extends uniquely to Killing vector
\begin{align}
X_{can}  & = \frac{1}{2} \Big(\aa + \Om^2 + z_1^2  - \sum\limits_{i=1}^{2p}x_i^2 \Big) \partial_{z_1} + z_1 \big(\Om  \partial_\Om  + \sum\limits_{i=1}^{2p}x_i \partial_{x_i} \big)   + \sum_{i=1}^p \mu_i ( x_{2i-1} \partial_{x_{2i}} - x_{2i} \partial_{x_{2i-1}}).
\end{align}
 Furthermore, $\{z_2 = 0\}$ is a totally geodesic submanifold of even-dimensional de Sitter spacetime.  Thus, the restrictions of the even-dimensional horizons to the $\{z_2 = 0 \}$ hypersurface are precisely the intrinsic Killing horizons of the odd-dimensional case. Namely, the parameters satisfy
 \begin{equation}\label{paramswhorizonodd}
  \aa \leq 0,\qquad \mu_1 = \cdots = \mu_p = 0,
 \end{equation}
and the horizons sit at
 \begin{equation}\label{horizonsreducedodd}
 \Om^2 = (z_1 \pm |\aa|^{1/2})^2  +\sum_i x_i^2.
\end{equation}
Additionally, by the same arguments used above, one easily verifies that $|\aa| \neq 0$ is a bifurcate horizon with bifurcation surface $\S^{n-1}$ and $|\aa| = 0$ is non-bifurcate.

\bigskip

As pointed out in Lemma \ref{lemmaequivKIDdS}, the results in this section depend only on the conformal class $[\xi]$ and not on the particular representative $\xi$ employed.  Moreover, by the conformal invariance of the asymptotic initial value problem (discussed after Theorem \ref{theoremavp} for conformally flat boundaries), the results must also be independent of the conformal extension of de Sitter spacetime. Hence, the classification obtained above can be straightforwardly stated in a conformally invariant way  (see Theorem \ref{theosuff} below.)

In addition, the results of this section construct all the Killing horizons of de Sitter spacetime from $\scri^+$, thus, it is a priori assumed  that  $\closure{\mathcal H_X}\cap\mathscr I^+\neq\emptyset$. However, in de Sitter spacetime every Killing horizon necessarily intersects $\mathscr I^+$ (e.g.~\cite{gallowaydS1,gallowaydS2}). Therefore the result is a \emph{complete characterization of all the Killing  horizons in de Sitter}.

\bigskip

The main result is as follows:

 \begin{theorem}\label{theosuff}
  Consider $(n+1)$-dimensional de Sitter spacetime $(\tilde M,\tilde g)$ and let $X$ be a Killing thereof, where $\xi = X\mid_{\scri^+}$. Assume that the conformal class $[\xi]$ is characterized by the parameters in \eqref{paramswhorizoneven} or \eqref{paramswhorizonodd}, according to the parity of $n$.  Then and only then $X$ has a Killing horizon $\mathcal{H}_X$. Moreover :
  \begin{enumerate}[a)]
   \item   If $|\aa|\neq 0$, $\mathcal{H}_X$ is bifurcate, thus non-degenerate, $\xi$ is conformal to a homothety and has exactly two essential isolated zeros, which coincide with $\closure{\mathcal{H}_X}\cap \scri^+$.
   \item   If $|\aa|= 0$, $\mathcal{H}_X$ is non-bifurcate, thus degenerate, $\xi$ is conformal to a translation and has exactly one essential isolated zero, which coincide with $\closure{\mathcal{H}_X}\cap \scri^+$.
  \end{enumerate}
 \end{theorem}

 \section{An asymptotic formula for the surface gravity}\label{secgrav}

In this section, we derive a formula for the surface gravity of a Killing horizon intersecting $\scri^+$ in terms of intrinsic quantities on the conformal boundary. Although Theorem~\ref{theoremnecessary} establishes that such an intersection can only occur in de Sitter spacetime, we perform the derivation (Subsection \ref{secproof}) for a general $\Lambda > 0$ vacuum spacetime admitting a conformal extension, as the general tensor calculation involves no additional complexity.

The main result is summarized in Proposition~\ref{propsurfacegrav}. We then evaluate its consequences for de Sitter spacetime in view of the classification established in Section~\ref{secsufficient}. The detailed step-by-step derivation of the formula is presented in Subsection~\ref{secproof}.

 \begin{proposition}\label{propsurfacegrav}
Let $(\tilde M, \tilde g)$ be a $(\Lambda>0)$-vacuum spacetime, admitting a conformal extension and a Killing vector $X$ with $\xi = X \mid_{\scri^+}$. Suppose that $X$ admits a Killing horizon $\mathcal{H}_X$ whose closure intersects $\scri^+$. Then the surface gravity of $\mathcal{H}_X$ is given at  $\closure{\mathcal{H}_X} \cap \scri^+$ by the conformal invariant
\begin{align}\label{surfagravmain}
\kappa^2 =\left[\left(\frac{\dvg_\gamma \xi}{n}\right)^2 - \frac{1}{2} |\skwxi|_\gamma^2\right]_{\closure{\mathcal{H}_X} \cap \scri^+},
\end{align}
where  $\skwxi_\gamma = D_{[A}\xi_{B]}$, $|\skwxi_\gamma|^2 = D_{[A}\xi_{B]} D^{[A}\xi^{B]}$ and $D$ is the Levi-Civita connection of $\gamma$.
\end{proposition}

From Theorem \ref{theoremnecessary}, we know that the hypotheses of Proposition \ref{propsurfacegrav} can only take place in (a locally) de Sitter spacetime. Moreover, by Lemma \ref{lemmaequivKIDdS}, the analysis can be carried out with any representative of the conformal class $[\xi]$, so we choose the canonical representative $\xi_{can}$. The result is identical for any parity of $n$, so we give the explicit details only the $n$ even case.

By Theorem \ref{theosuff}, $\xi_{can}$ extends to a Killing vector admitting a horizon only if it is given by \eqref{xicanhorizonsneven}. The differential of the associated covector $\bm \xi_{can} = \gamma(\xi_{can},\cdot)$ is simply $d \bm\xi_{can} =  0$. Thus,
 taking into account that $2 D_{[A}\xi_{B]} = (d \bm\xi_{can})_{AB}$ we obtain $|\skwxi_\gamma|^2 = 0$.
On the other hand, the value of $\dvg_\gamma \xi$ was calculated in equation \eqref{dvgxican} above.
Therefore
\begin{equation}
 \left(\frac{\dvg_\gamma \xi}{n}\right)^2 - \frac{1}{2} |\skwxi|_\gamma^2 = z_1^2.
\end{equation}
Now to obtain the surface gravity we just have to evaluate this expression at $\closure{\mathcal{H}_X} \cap \scri^+$, which from Theorem \ref{theosuff}, we know it coincides with the  essential isolated zeros of $\xi_{can}$. From Lemma \ref{zerosetxican} and equation \eqref{paramswhorizoneven} we have
\begin{equation}
 \kappa^2 = |\aa|.
\end{equation}
Thus, from the results of the Section \ref{secsufficient}, $\aa\neq 0$ corresponds to a bifurcate horizon whose two branches converge, yielding a non-bifurcate horizon as $\aa \to 0$. Therefore, this recovers the well-known result that bifurcate horizons in de Sitter are non-degenerate, while non-bifurcate ones are degenerate.

 \subsection{Proof of Proposition \ref{propsurfacegrav}}\label{secproof}

 Let $X$ be a Killing vector field with a Killing horizon $\mathcal H_X$. The surface gravity $\kappa \in \R$, in terms of the physical metric $\tilde g$, is defined by
 \begin{equation}
  X^a \tilde \nabla_a X^b \mid_{\mathcal H_X} = \kappa X^b.
 \end{equation}
A well-known formula (e.g. \cite{wald}) allows to obtain the square of $\kappa$ by computing
\begin{equation}
 \kappa^2 = -\frac{1}{2}\tilde \nabla_a \tilde X_b \tilde \nabla^a  X^b,
\end{equation}
where $\tilde X_b = \tilde g_{ab} X^a$. Our aim is to rewrite this in terms of regular quantities at $\scri^+$.

We define $ X_b =  g_{ab} X^a$, which satisfies $X_b = \Om^2 \tilde X_b,$ and recall the change of connection tensor
\begin{equation}
 \left(\tilde \nabla - \nabla \right)^c{}_{ab} = : S^c{}_{ab}
 = -\frac{1}{\Om} \left( \nabla_a \Om \delta^c{}_b + \nabla_b \delta^c{}_a - \nabla_d \Om g^{cd} g_{ab} \right),
\end{equation}
which yields
\begin{align}
\tilde \nabla_a \tilde X_b  & = \nabla_a \tilde X_b - S^c{}_{ab} \tilde X_c= \nabla_a \tilde X_b + \frac{1}{\Om} \left( \nabla_a \Om \delta^c{}_b + \nabla_b \delta^c{}_a - \nabla_d \Om g^{cd} g_{ab} \right)\tilde X_c \\
& = \frac{\nabla_a  X_b}{\Om^2} + \frac{1}{\Om^3} \left(  \nabla_b \Om X_a -\nabla_a \Om X_b  - X(\Om) g_{ab} \right),\label{changeconnabX}
\end{align}
where we have written $X(\Om) = X^d \nabla_d \Om $. Now observing that
\begin{equation}
 \tilde \nabla_a \tilde X_b \tilde \nabla^a \tilde X^b = \tilde g^{ac} \tilde g^{bd}\tilde \nabla_a \tilde X_b \tilde \nabla_c \tilde X_d =  \Om^4  g^{ac}  g^{bd}\tilde \nabla_a \tilde X_b \tilde \nabla_c \tilde X_d
\end{equation}
and denoting
\begin{equation}
  \tilde \nabla_a \tilde X_b  \tilde \nabla^a \tilde X^b \equiv |\tilde \nabla \tilde X|_{\tilde g}^2 \qquad \mbox{and} \qquad  \nabla_a  X_b   \nabla^a  X^b \equiv | \nabla  X|_{g}^2,
\end{equation}
we have
\begin{align}
   |\tilde \nabla \tilde X|_{\tilde g}^2
= &  \left( \nabla_a  X_b + \frac{1}{\Om} \left(  \nabla_b \Om X_a -\nabla_a \Om X_b  - X(\Om) g_{ab} \right)\right) \left( \nabla^a  X^b + \frac{1}{\Om} \left(  \nabla^b \Om X^a -\nabla^a \Om X^b  - X(\Om) g^{ab} \right)\right)\\
= &| \nabla  X|_{g}^2 + \frac{2}{\Om} \nabla_a  X_b \left(  \nabla^b \Om X^a -\nabla^a \Om X^b  - X(\Om) g^{ab} \right) + \frac{1}{\Om^2}\left(2 |X|_g^2 |\nabla \Om|_g^2 + (n-1) X(\Om)^2 \right).
\end{align}
Since $X^a$ is a Killing vector of $\tilde g$ and a CKV of $g$, as we have deduced in equation \eqref{XOmdivX}, we have
\begin{equation}\label{XOmdivX2}
 \quad X(\Om)  = \Om \phi = \frac{\Om}{n+1} \dvg_g X,
\end{equation}
thus
\begin{align}\label{eq1}
 &  |\tilde \nabla \tilde X|_{\tilde g}^2 =   | \nabla  X|_{g}^2  - \frac{n+3}{(n+1)^2} (\dvg_g X)^2 + \frac{2}{\Om} \nabla_a  X_b \left(  \nabla^b \Om X^a -\nabla^a \Om X^b  \right) + \frac{2}{\Om^2} |X|_g^2 |\nabla \Om|_g^2.
\end{align}
The $\Om^{-2}$ term in the RHS diverges as one approaches to $\scri^+$. However, this term does not contribute as one approaches $\scri^+$ through the horizon. We collect the terms that are constant (or, in particular, vanish) on $\mathcal{H}_X$ in the LHS of the equation
\begin{equation}
|\tilde \nabla \tilde X|_{\tilde g}^2 - \frac{2}{\Om^2} |X|_g^2 |\nabla \Om|_g^2 = | \nabla  X|_{g}^2  - \frac{n+3}{(n+1)^2} (\dvg_g X)^2 + \frac{2}{\Om} \nabla_a  X_b \left(  \nabla^b \Om X^a -\nabla^a \Om X^b  \right).\label{surfagrav1}
\end{equation}
We now keep working the $\Om^{-1}$ term on the RHS. First, using the conformal Killing equation of $X$ wrt $g$ as well as \eqref{XOmdivX2}, we have
\begin{equation}
 \nabla_a X_b \left(  \nabla^b \Om X^a -\nabla^a \Om X^b  \right) = 2 \nabla_a X_b   \nabla^b \Om X^a -\frac{2 X(\Om)}{n+1} \dvg_g X = 2 \nabla_a X_b   \nabla^b \Om X^a -\frac{2 \Om}{(n+1)^2} (\dvg_g X)^2,\label{NXNOm0}
\end{equation}
where the first term can be in turn writte
\begin{align}
 2 \nabla_a X_b \nabla^b \Om X^a & = 2 X^a \nabla_a X(\Om) -2 X^a X^b \nabla_a \nabla_b \Om \\ & =  \frac{2 X(\Om)}{n+1} \dvg_g X + \frac{2 \Om}{n+1} X^a \nabla_a \dvg_g X-2 X^a X^b \nabla_a \nabla_b \Om \\
 & = \frac{2 \Om}{(n+1)^2} (\dvg_g X)^2 + \frac{2 \Om}{n+1} X^a \nabla_a \dvg_g X-2 X^a X^b \nabla_a \nabla_b \Om.\label{NXNOm1}
\end{align}
From the quasi-Einstein equation \eqref{quasiEinstein}
we obtain that the trace-free part of $\nabla_a \nabla_b \Om$, so the decomposition into trace and trace-free part yields
\begin{equation}\label{hessOm}
 \nabla_a \nabla_b \Om =-\frac{\Om}{n-1} R^{(tf)}_{ab} + \frac{\Box \Om}{n+1} g_{ab}
\end{equation}
being $R^{(tf)}_{ab}$ the trace-free part of the Ricci tensor $R_{ab}$. Hence, we write \eqref{NXNOm1} as
\begin{align}
 2 \nabla_a X_b \nabla^b \Om X^a & = \frac{2 \Om}{(n+1)^2} (\dvg_g X)^2 + \frac{2 \Om}{n+1} X^a \nabla_a \dvg_g X + \frac{2 \Om}{n-1}R^{(tf)}_{ab} X^a X^b - \frac{2 \Box \Om}{n+1} |X|^2_g .
\end{align}
Thus, \eqref{NXNOm0} becomes
\begin{align}
 \nabla_a X_b \left(  \nabla^b \Om X^a -\nabla^a \Om X^b  \right) = \frac{2 \Om}{n+1} X^a \nabla_a \dvg_g X + \frac{2 \Om}{n-1}R^{(tf)}_{ab} X^a X^b - \frac{2 \Box \Om}{n+1} |X|^2_g
\end{align}
and \eqref{surfagrav1} yields
\begin{align}
& |\tilde \nabla \tilde X|_{\tilde g}^2 - \frac{2}{\Om^2} |X|_g^2 |\nabla \Om|_g^2 + \frac{4 \Box \Om}{n+1} \frac{|X|^2_g}{\Om}\\ = & | \nabla  X|_{g}^2  - \frac{n+3}{(n+1)^2} (\dvg_g X)^2 +  \frac{4}{n+1} X^a \nabla_a \dvg_g X + \frac{4}{n-1}R^{(tf)}_{ab} X^a X^b\label{surfagrav2}
\end{align}
where $\frac{4 \Box \Om}{n+1} \frac{|X|^2_g}{\Om}$ has been placed on the LHS because it also vanishes on $\mathcal{H}_X$.

\bigskip

The RHS of equation \eqref{surfagrav2} is already explicitly regular at $\scri^+$. Indeed, the LHS equals $-2 \kappa^2$ on  $\mathcal{H}_X$, thus, approaching $\scri^+$ through $\mathcal{H}_X$, the RHS already gives a formula to calculate the surface gravity at $\scri^+$.

However, one would like to have a formula written solely in terms of $\xi = X\mid_{\scri^+}$, its  derivatives and intrinsic objects of $\scri^+$. For that, consider the orthogonal splitting
\begin{equation}
 g_{ab} = \frac{1}{|\nabla \Om|_ g^2}\nabla_a \Om \nabla_b \Om + \gamma_{ab}
\end{equation}
where $\gamma_{ab}$ is the projector orthogonal  to $\nabla_a \Om$. The Levi-Civita derivative with respect to $\gamma$ will be denoted $D$. Then, by \eqref{XOmdivX2}, $X$ is written
\begin{equation}\label{normtanX}
 X_a = X(\Om) \frac{\nabla_a \Om}{|\nabla \Om|_g^2}+ \xi_a = \Om\frac{\dvg_g X}{n+1} \frac{\nabla_a \Om}{|\nabla \Om|_g^2} + \xi_a,\qquad  \xi_a \gamma^a{}_b = \xi_b.
\end{equation}
For simplicity, in the remainder the terms that are clearly of order $O(\Om)$ will not be explicity displayed in the calculations because they do not contribute at $\scri^+$ and will be collected within the $O(\Om)$ symbol.

First from \eqref{normtanX} it follows that
\begin{align}\label{asymnabX}
 \nabla_a X_b = \frac{\dvg_g X}{n+1} \frac{1}{|\nabla \Om|_g^4}  \nabla_a \Om \nabla_b \Om + \nabla_a \xi_b + O(\Om).
\end{align}
so the conformal Killing equation reads
\begin{align}
 & 2\nabla_{(a} X_{b)}  = 2 \frac{\dvg_g X}{n+1} \frac{1}{|\nabla \Om|_g^4}  \nabla_a \Om \nabla_b \Om + 2 \nabla_{(a} \xi_{b)} + O(\Om) = 2\frac{\dvg_g X}{n+1} g_{ab}\\
 \Longrightarrow & \frac{2}{n+1} \dvg_g X \left(g_{ab}- \frac{1}{|\nabla \Om|_g^4}  \nabla_a \Om \nabla_b \Om \right)  = 2 \nabla_{(a} \xi_{b)} + O(\Om).\label{asymptCKeq}
\end{align}
By observing now that the term $\nabla_a \xi_b$ has normal-normal component
\begin{align}\label{nornornabxi}
  \nabla^c \Om \nabla^d \Om \nabla_c \xi_d = - \nabla^c \Om \xi^d \nabla_c \nabla_d \Om = O(\Om),
\end{align}
 the trace of \eqref{asymptCKeq} gives
\begin{equation}\label{divXdivxi}
\dvg_g X = \frac{n+1}{n} \dvg_\gamma \xi + O(\Om).
\end{equation}
On the other hand, the normal-tangent components of $\nabla_a \xi_b$ are
\begin{align}
& \frac{1}{|\nabla \Om |_g^2}\left( \nabla_a \Om \nabla^c \Om \gamma^d{}_b \nabla_c \xi_d +\nabla_b \Om \nabla^d \Om \gamma^c{}_a \nabla_c \xi_d \right) = \frac{1}{|\nabla \Om |_g^2}\left( \nabla_a \Om \nabla^c \Om \gamma^d{}_b \nabla_c \xi_d -\nabla_b \Om \xi^d \gamma^c{}_a \nabla_c \nabla_d \Om \right)
 \label{nortanxi1}.
\end{align}
By \eqref{asymptCKeq} if follows
\begin{equation}
 \nabla^c \Om \gamma^d{}_b \nabla_c \xi_d = - \nabla^c \Om \gamma^d{}_b \nabla_d \xi_c + O(\Om),
\end{equation}
thus, \eqref{nortanxi1} yields
\begin{align}
& \frac{1}{|\nabla \Om |_g^2}\left( -\nabla_a \Om \nabla^c \Om \gamma^d{}_b \nabla_d \xi_c -\nabla_b \Om \xi^d \gamma^c{}_a \nabla_c\nabla_d \Om  \right) + O(\Om) \\
 = & \frac{1}{|\nabla \Om |_g^2}\left( \nabla_a \Om  \xi^c \gamma^d{}_b \nabla_d \nabla_c \Om   -\nabla_b \Om \xi^d \gamma^c{}_a \nabla_c\nabla_d \Om  \right) + O(\Om)  \\
  = &  \frac{2}{|\nabla \Om |_g^2} \nabla_{[a} \Om \gamma^c{}_{b]} \xi^d \nabla_c \nabla_d \Om  + O(\Om),
\end{align}
Therefore, taking into account \eqref{nornornabxi}, we get the decomposition of $\nabla_a \xi_b$ into normal and tangent components
\begin{align}
 \nabla_a \xi_b = \frac{2}{|\nabla \Om |_g^2} \nabla_{[a} \Om \gamma^c{}_{b]} \xi^d \nabla_c \nabla_d \Om  + D_a \xi_b + O(\Om)
\end{align}
and putting it into \eqref{asymnabX}, together with \eqref{divXdivxi},  we obtain
\begin{equation}\label{decomnabX}
 \nabla_a X_b = \frac{\dvg_\gamma \xi}{n} \frac{1}{|\nabla \Om|_g^4}  \nabla_a \Om \nabla_b \Om +\frac{2}{|\nabla \Om |_g^2} \nabla_{[a} \Om \gamma^c{}_{b]} \xi^d \nabla_c \nabla_d \Om  + D_a \xi_b + O(\Om).
\end{equation}
Hence, from this expression,
\begin{align}
 |\nabla X|_g^2 & = \left(\frac{\dvg_\gamma \xi}{n}\right)^2 + |D \xi|_\gamma^2 + \frac{4}{|\nabla \Om |_g^4} \nabla_{[a} \Om \gamma^c{}_{b]}  \nabla^{[a} \Om \gamma^{|c'|b]} \xi^d \xi^{d'} \nabla_c \nabla_d \Om  \nabla_{c'} \nabla_{d'} \Om  + O(\Om) \\
 & =  \left(\frac{\dvg_\gamma \xi }{n}\right)^2 + |D \xi|_\gamma^2 + \frac{2}{|\nabla \Om |_g^2} \gamma^{cc'} \xi^d \xi^{d'} \nabla_c \nabla_d \Om  \nabla_{c'} \nabla_{d'} \Om  + O(\Om) \\
 & =   \left(\frac{\dvg_\gamma \xi}{n}\right)^2 + |D \xi|_\gamma^2 + \frac{2}{|\nabla \Om |_g^2} \left(\frac{\Box \Om}{n+1}\right)^2 |\xi|_\gamma^2 + O(\Om)\label{normnabX}
\end{align}
where for the last equality we have made use of \eqref{hessOm}. As a final step, we decompose $D_a \xi_b$ into sum of skew and symmetric parts
\begin{equation}
 D_{a} \xi_b =  D_{[a}\xi_{b]} +  D_{(a}\xi_{b)} =\skwxi_{ab}  + \frac{\dvg_\gamma \xi}{n} \gamma_{ab} + O(\Om),\qquad \skwxi_{ab} := D_{[a}\xi_{b]},
\end{equation}
where we have made use of the ``asymptotic'' conformal Killing equation
\begin{equation}\label{asymCKVeqxi}
 2 D_{(a}\xi_{b)} = \frac{2}{n} \dvg_\gamma \xi + O(\Om),
\end{equation}
which follows from inserting decomposition \eqref{decomnabX} together with \eqref{divXdivxi} in the conformal Killing equation of $X$. Thus, \eqref{normnabX} becomes
\begin{align}
 |\nabla X|_g^2 &  =  (1+n) \left(\frac{\dvg_\gamma \xi}{n}\right)^2 + |\skwxi|_\gamma^2 + \frac{2}{|\nabla \Om |_g^2} \left(\frac{\Box \Om}{n+1}\right)^2 |\xi|_\gamma^2 + O(\Om).\label{nabX2}
\end{align}
Expression \eqref{nabX2} gives the first term in the RHS of \eqref{surfagrav2} in terms of $\xi$, as we wanted. For the last two terms in the RHS of \eqref{surfagrav2}, we substitute $X$ by the decomposition in \eqref{normtanX} and $\dvg_g X$ by \eqref{divXdivxi}
\begin{align}
 \frac{4}{n+1} X^a \nabla_a \dvg_g X + \frac{4}{n-1}R^{(tf)}_{ab} X^a X^b = \frac{4}{n} \xi^a D_a \dvg_\gamma \xi + \frac{4}{n-1}R^{(tf)}_{ab} \xi^a \xi^b + O(\Om).\label{tailX}
\end{align}

In summary, renaming
\begin{equation}
 \sgrav : =  |\tilde \nabla \tilde X|_{\tilde g}^2 - \frac{2}{\Om^2} |X|_g^2 |\nabla \Om|_g^2 + \frac{4 \Box \Om}{n+1} \frac{|X|^2_g}{\Om},
\end{equation}
and inserting \eqref{divXdivxi},\eqref{nabX2} and \eqref{tailX} into \eqref{surfagrav2}, one gets
\begin{align}
 \sgrav  = |\skwxi|_\gamma^2 -2  \left(\frac{\dvg_\gamma \xi}{n}\right)^2 +   \frac{2}{|\nabla \Om |_g^2} \left(\frac{\Box \Om}{n+1}\right)^2 |\xi|_\gamma^2  + \frac{4}{n} \xi^a D_a \dvg_\gamma \xi + \frac{4}{n-1}R^{(tf)}_{ab} \xi^a \xi^b + O(\Om).\label{surfagrav2.1}
\end{align}
Equation \eqref{surfagrav2.1} shows that the function $\mathcal{G}$ extends to $\scri^+$ and, by definition, it satisfies $\mathcal{G}\mid_{\mathcal{H}_X} = - 2 \kappa^2$. Hence, at $\closure{\mathcal{H}_X} \cap \scri^+ $, which is a subset of the zero set of $\xi$, $\zeroset{\xi}$, one has
\begin{align}\label{surfagrav3}
\kappa^2 =\left[\left(\frac{\dvg_\gamma \xi}{n}\right)^2 - \frac{1}{2} |\skwxi|_\gamma^2\right]_{\closure{\mathcal{H}_X} \cap \scri^+}.
\end{align}
It is interesting to evaluate the conformal behaviour of the RHS of \eqref{surfagrav3}. Under a conformal scaling $\hat \gamma = \omega^2 \gamma$, which we assume regular at $\zeroset{\xi}$, and denoting $\hat\xi_a = \hat \gamma_{ab} \xi^b$, we have by an analogous argument as the one used in equation \eqref{changeconnabX},
\begin{equation}
 D_a \xi_b = \hat D_a \hat \xi_b + \frac{1}{\omega}\left(2 \hat D_{[a}\omega \hat \xi_{b]} - \xi(\omega) \hat \gamma_{ab} \right).
\end{equation}
Thus, it is clear that the expression in the RHS of \eqref{surfagrav3} it is not everywhere an invariant under conformal scaling. However, at $\closure{\mathcal{H}_X} \cap \scri^+ \subset \zeroset{\xi}$ it does hold $ D_a \xi_b = \hat D_a \hat \xi_b$, hence the expresion \eqref{surfagrav3} for the surface gravity is a conformal invariant
\begin{align}\label{surfagrav3}
\kappa^2 =\left[\left(\frac{\dvg_\gamma \xi}{n}\right)^2 - \frac{1}{2} |\skwxi|_\gamma^2\right]_{\closure{\mathcal{H}_X} \cap \scri^+}= \left[\left(\frac{\dvg_{\hat\gamma} \xi}{n}\right)^2 - \frac{1}{2} |\hat \skwxi|_{\hat \gamma}^2\right]_{\closure{\mathcal{H}_X} \cap \scri^+}.
\end{align}

 \section{Discussion}\label{secdiscussion}

 In this paper, we have established a local characterization of $(n+1)$-dimensional de Sitter as the unique $(\Lambda > 0)$-vacuum spacetime admitting a Killing vector $X$ with a Killing horizon intersecting the conformal boundary $\mathscr{I}^+$. This intersection occurs precisely at the essential isolated zeros of the CKV $\xi = X|_{\mathscr{I}^+}$ and the characterization holds in a neighbourhood of these points. If $\xi$ is complete, the characterization turns out to be global.

 Additionally, based on the above result,  an asymptotic characterization and classification of Killing horizons in $(n+1)$-dimensional de Sitter has been obtained. By systematically reconstructing the bulk Killing vectors from their asymptotic boundary data, we showed that a boundary CKV $\xi$ extends to a Killing vector admitting a horizon if and only its conformal class  $[\xi]$ (see equation \eqref{confclass}) belongs to the subset parametrized by the values of the conformal invariants (see Appendix \ref{appcanonical})
 \begin{equation}
  \aa\leq 0,~\bb = \mu_1 = \dots = \mu_p = 0,\quad\mbox{if $n$ even},\quad\cc\leq 0,~ \mu_1 = \dots = \mu_p = 0,\quad\mbox{if $n$ odd}.
 \end{equation}
 This classifies the resulting horizons into non-degenerate bifurcate horizons (if  $\sigma < 0$) and degenerate non-bifurcate (if $\sigma = 0$).  A  conformally invariant formula expressing the surface gravity of these cases as $\kappa^2 = |\aa|$ has also been proven.

 \bigskip

 Beyond providing a complete asymptotic classification for pure de Sitter Killing horizons, these results are an essential building block for establishing a framework for the analysis of the asymptotic behavior of cosmological horizons in more general spacetimes with a positive cosmological constant. For simplicity, we restrict the following discussion to four spacetime dimensions.

Recall from  Section~\ref{secdata} that in four dimensions the asymptotic data consist of a Riemannian $3$-manifold $(\Sigma^3,\gamma)$ endowed with a traceless and transverse tensor $g_{(3)}$, respectively prescribing the zero and third order of the Fefferman-Graham expansion (cf. \cite{FeffGrah85,ambientmetric,Starob82}). In turn,
the manifold $(\Sigma^3,\gamma)$ prescribes the geometry of $\scri$, while $g_{(3)}$ is the electric part of the rescaled Weyl tensor, the latter characterization being  generally true only in four dimensions (cf. \cite{friedrich81,friedrich81bis,Fried86initvalue} and \cite{marspeondata21}.) As a result of the Fefferman-Graham asymptotic expansion, two metrics $g$ and $g'$ characterized by respective data $(\Sigma^3,\gamma,g_{(3)})$ and $(\Sigma^3,\gamma,g_{(3)}')$ coincide up to third order.

 This allows to consider a fixed background spacetime $(M,g)$, correspoding to data $(\Sigma^3,\gamma,g_{(3)} = 0)$. At the conformal boundary of \emph{the background spacetime}, we can prescribe more general data $(\Sigma'^3,\gamma,g_{(3)}')$, where $\Sigma'^3$ is the maximal domain $\Sigma'^3 \subseteq \Sigma^3$ where $g_{(3)}'$ is regular. The local evolution of the data $(\Sigma'^3,\gamma,g_{(3)}')$ yields a maximal Cauchy development $(M',g')$ with $M' \subseteq M$ (cf. Figure \ref{figbackground}.)
 In this context, de Sitter spacetime appears not just as another solution of the Einstein equations, but as a general background manifold for spacetimes with locally conformally flat $\scri$.

 \begin{figure}[ht]
\centering
\begin{tikzpicture}[scale=1.2]
\coordinate (TL) at (0,4);
\coordinate (TR) at (4,4);
\coordinate (BL) at (0,0);
\coordinate (BR) at (4,0);
\coordinate (C)  at (2,2);

\draw[thick] (TL)--(TR);
\draw[thick] (TL)--(BL);
\draw[thick] (TR)--(BR);
\draw[thick] (BL)--(BR);

\fill[pattern={
Lines[
angle=45,
distance=7.5pt,
line width=0.3pt
]},pattern color=black!80,]
(TL)--(BL)--(BR)--(TR)--cycle;

\fill[pattern={
Lines[
angle=-45,
distance=7.5pt,
line width=0.3pt
]},pattern color=red]
(TL)--(TR)--(C)--cycle;

\draw[thick] (TL)--(C)--(BR);
\draw[thick] (BL)--(C)--(TR);


\fill (TL) circle (2pt);
\fill (TR) circle (2pt);

\node[above] at ($(TL)!0.26!(TR)$)
{$(\Sigma^3,\gamma,g_{(3)}=0)$};

\node[above,text=red] at ($(TL)!0.75!(TR)$)
{$(\Sigma^3,\gamma,g_{(3)}')$};

\node[left]  at (TL) {$\mathcal H_1\cap\mathscr I^{+}$};
\node[right] at (TR) {$\mathcal H_2\cap\mathscr I^{+}$};

\node at (1.2,2.2) {$\mathcal H_1$};
\node at (2.8,2.2) {$\mathcal H_2$};


\end{tikzpicture}
\caption{\small\label{figbackground} Penrose diagram of a bifurcate horizon intersecting $\scri^+$ in background de Sitter (black lines). Another spacetime metric (red lines) can be propagated from asymptotic data in  the domain of dependence of $\Sigma^3-\{\mathcal{H}_1 \cup\mathcal{H}_2 \cap \scri^+ \}$.}
\end{figure}
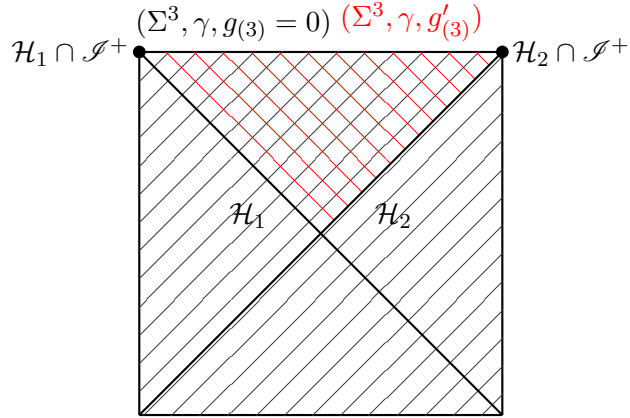

The main interest of the above construction lies in the fact that the background manifold provides a topology, that of $M$, containing points missing in $M'$. For instance, cosmological horizons $\mathcal{H}_X$ often appear as Cauchy horizons of the domain of dependence of $\scri^+$. Consequently, the intersection $\closure{\mathcal{H}_X} \cap \scri^+$ is empty with respect to the topology of $M'$, but it may exist as a limit point of $M'$ with respect to the topology of $M$. Thus, this machinery is well suited to address whether the presence of such a horizon appers as a pathological behavior in the asymptotic data of $g'$, particularly in $g_{(3)}'$.

The situation described above arises in physically important spacetimes, the simplest being the Schwarzschild-de Sitter spacetime. Its asymptotic data consists of a conformally flat manifold $(\Sigma^3,\gamma)$ together with
\begin{equation}\label{dataKdSLike}
  g_{(3)}' = |\xi|^{-5} (\xi \otimes \xi)^{\mathrm{tf}},
\end{equation}
where $\xi$ is a CKV \emph{conformally equivalent to a homothety} and the superscript $\mathrm{tf}$ denotes the trace-free part. The tensor $g_{(3)}'$ satisfies the KID equation \eqref{eqKID} with respect to $\xi$ and extends to the Killing generator $X$ of a cosmological (as well as a black hole) horizon $\mathcal{H}_X$. The intersection $\overline{\mathcal{H}_X} \cap \scri^+$ is empty with respect to the Schwarzschild-de Sitter topology, but non-empty with respect to the background de Sitter topology, occurring exactly at the essential isolated zeros of $\xi$. Thus, $g_{(3)}'$ diverges there, exhibiting a pathological behavior as conjectured in the previous paragraph.

In addition, the more general family of Kerr-de Sitter black holes, and, more generally, the so-called Kerr-de Sitter-like class (cf.~\cite{KdSnullinfty,KdSlike}), have data of the form \eqref{dataKdSLike} with  $(\Sigma^3,\gamma)$  conformally flat, with the difference that $\xi$ lies in a different conformal class. This  suggests that the framework discussed here may prove useful in the study of cosmological black hole uniqueness.

\section*{Acknowledgements}
The author is grateful to Marc Mars for reviewing this manuscript, as well as Igor Khavkine and Alfonso García-Parrado for helpful input and ideas. The author acknowledges financial support under the project PID2024-155175NB-I00 (Agencia Estatal de Investigación, Ministerio de Ciencia, Innovación y Universidades) and Departamento de Matem\'atica Aplicada a las TIC-D540.

 \appendix

 \section{Classification of conformal Killing vector fields in $\mathbb{S}^n$}\label{appcanonical}

In this section, we review the classification of CKVs modulo conformal isometries in conformally flat manifolds. We follow closely the results in \cite{marspeon21.1, marspeon21}, where the classification is achieved by constructing a unified canonical form comprising all equivalence classes. Other classifications are possible (see e.g. \cite{KdSlike} or \cite{djokovic83} for a reference in a more algebraic language), but the one in this Appendix suits our purposes best.

We consider the $n$-sphere $\mathbb{S}^n$, respresented by the one-point compatification $\mathbb{S}^n = \mathbb{E}^n \cup \{ \infty \}$, as our model space. We denote the conformal group of $\S^n$ (resp. $\mathbb{E}^n$) by $\mathrm{Conf}(\mathbb{S}^n)$ (resp. $\mathrm{Conf}(\mathbb{E}^n)$) and the Lie algebra of CKVs $\mathrm{CKill}(\mathbb{S}^n)$ (resp. $\mathrm{CKill}(\mathbb{E}^n)$). The classification of $\mathrm{CKill}(\mathbb{S}^n)/\mathrm{Conf}(\mathbb{S}^n)$ is globally well-defined and maps equivalently to the local classification in any locally conformally flat manifold (cf. \cite{marspeonKSKdS21}.)

An elementary result in conformal geometry (cf. \cite{IntroCFTschBook}) is that the Lie group $\mathrm{Conf}(\mathbb{S}^n)$ is isomorphic to the group of isotropies $O^+(1,n+1)$ of an ambient ($n+2$)-dimensional Minkowski space $\mathbb{M}^{1,n+1}$, also known as the orthochronous component of the Lorentz group.
From the isomorphism between the Lie groups, it follows the isomorphism between the Lie algebras $ \mathfrak{o}(1,n+1)$ and $ \mathrm{CKill}(\mathbb{S}^n)$. We represent the elements of the former as skew-symmetric endomorphisms of $\mathbb{M}^{1,n+1}$, $\skwend{\mathbb{M}^{1,n+1}}$. The explicit form of the isomorphism is given in the conformally equivalent representation of $\S^n$ as the Euclidean space $\mathbb{E}^n$

We write an element $F \in \skwend{\mathbb{M}^{1,n+1}}$ with respect to an orthonormal basis of $\mathbb{M}^{1,n+1}$ as
\begin{equation}\label{skwmatrix}
 F = \begin{pmatrix}
  0 &  -\nu & -\a^t + \b^t/2 \\
  -\nu & 0 & - \a^t - \b^t/2 \\
  -\a + \b/2 & \a +\b/2 & -\pmb{\omega}
  \end{pmatrix},
\end{equation}
with $\nu \in \R$, $\a,\b \in \R^n$ row vectors, and $\pmb{\omega}$ an $n \times n$ skew matrix $\pmb{\omega}^t = - \pmb{\omega}$. Endowing $\mathbb{E}^n$ with Cartesian coordinates $\{ y^A\}_{A=1}^n$, with the choices in \cite{marspeon21}, \eqref{skwmatrix} maps  to
\begin{equation}\label{CKVFgeneral}
\xi_F=
\lr{
\b^A+\nu y^A+(\a_B y^B)y^A
-\frac12(y_B y^B)\a^A
-{\omega^A}_B y^B
}\partial_{y^A} \in \conf{\mathbb{E}^n}.
\end{equation}

As a consequence of this isomorphism, classifying CKVs modulo $\mathrm{Conf}(\mathbb{S}^n)$ is equivalent to classify skew-symmetric endomorphisms $F \in \mathfrak{o}(1,n+1)$ modulo the adjoint action of $O^+(1,n+1)$. In other words the space of \emph{adjoint orbits}
\begin{equation}
[F] = \{ F' \in\skwend{\mathbb{M}^{1,n+1}} \tq F' = \mathrm{Ad}_L(F) = L \cdot F \cdot L^{-1},\, \forall L \in O^+(1,n+1) \}
\end{equation}
is in one-to-one correspondence with the space of \emph{conformal classes}
\begin{equation}
 [\xi]=\{ \xi' \in \mathrm{CKill}(\S^n) \tq \xi' = \varphi_\star(\xi),\, \forall \varphi \in \mathrm{Conf}(\mathbb E^n)\}.
\end{equation}

The classification of both quotient spaces is achieved in \cite{marspeon21,marspeon21.1} by finding a canonical form representing each one of these classes. For the adjoint orbits $[F]$, the canonical representative is, depending on the parity of the dimension $n$, one of the following matrices:
\begin{align}
 F_{can} & = \begin{pmatrix}
  0 & 0 & -1 + \frac{\sigma}{4} & \frac{\tau}{4} \\
  0 & 0 & -1-\frac{\sigma}{4} & -\frac{\tau}{4} \\
  -1+\frac{\sigma}{4} & 1+\frac{\sigma}{4} & 0 & 0 \\
  \frac{\tau}{4} & \frac{\tau}{4} & 0 & 0
 \end{pmatrix} \bigoplus_{i = 1}^{p} \begin{pmatrix}
  0 & -\mu_i \\
  \mu_i & 0
 \end{pmatrix}\quad\quad \mbox{if $n$ even}, \label{decompFeven} \\
 F_{can} & = \begin{pmatrix}
   0 &  0 & -1 + \frac{\sigma}{4} \\
   0 & 0 & - 1 - \frac{\sigma}{4} \\
   -1 + \frac{\sigma}{4} & 1 + \frac{\sigma}{4} & 0
  \end{pmatrix} \bigoplus_{i = 1}^{p} \begin{pmatrix}
  0 & -\mu_i \\
  \mu_i & 0
 \end{pmatrix}\quad\quad\mbox{if $n$ odd}\label{decompFodd},
\end{align}
where
\begin{equation}
 p:=[(n+1)/2] -1.
\end{equation}
The parameters $\sigma, \tau, \mu_i \in \mathbb{R}$ are $O^+(1,n+1)$-invariants determined by the roots of the characteristic polynomial associated with\footnote{The minus sign is just a convention.} $-F^2$. First, because of skew-symmetry all roots are at least double and there is always a zero root when $n$ is odd, so in order to remove this trivial information from the spectrum we define
\begin{equation}
\mathcal{Q}(x):=\sqrt{P_{F^2}(-x)},\qquad \mbox{$n$ even,}
\qquad\qquad
\mathcal{Q}(x):=\sqrt{P_{F^2}(-x)/x},\qquad \mbox{$n$ odd.}
\end{equation}
The parameters are then extracted from the sorted roots of $\mathcal{Q}(x)$.
The ordering of the roots is based on a discrete invariant, the causal character of $\ker F$, that allows to distinguish adjoint orbits with identical spectra\footnote{Other discrete parameters may equivalently be used to distinguish orbits with identical spectra, for instance, the matrix rank of $F$.} in the following way:

\begin{definition}\label{defgammamu}
 Let $\mathrm{Roots}\lr{\mathcal{Q}_{F^2}}$ denote roots of $\mathcal{Q}_{F^2}(x)$ repeated as many times as their multiplicity. Then
 \begin{enumerate}
  \item[a)] If $n$ is odd, $\lrbrace{\sigma; \mu_1^2, \cdots, \mu_{p}^2} := \mathrm{Roots}\lr{\mathcal{Q}_{F^2}}$, sorted by $\sigma \geq \mu_1^2\geq \cdots \geq \mu_{p}^2$ if $\ker F$ is timelike, and $\mu_1^2\geq \cdots \geq \mu_{p}^2\geq 0 \geq \sigma$ otherwise.
  \item[b)] If $n$ is even, $\sigma := \mu_s^2 - \mu_t^2$ and $\tau :=2 |\mu_t \mu_s|$, where $\lrbrace{-\mu_t^2, \mu_s^2; \mu_1^2, \cdots, \mu_{p}^2} := \mathrm{Roots}\lr{\mathcal{Q}_{F^2}}$. These are sorted by $\mu_1^2\geq \cdots \geq \mu_{p}^2\geq\mu_s^2 = -\mu_t^2 = 0$ if $\ker F$ is degenerate, and $\mu_s^2 \geq \mu_1^2\geq \cdots \geq \mu_{p}^2\geq 0 \geq  -\mu_t^2$ otherwise.
 \end{enumerate}
\end{definition}

\begin{remark}
  From the discussion in this Appendix, observe that to calculate the parameters in Definition \ref{defgammamu} from a CKV in $\mathbb{E}^n$, one requires to find Cartesian coordinates $\mathbb{E}^n$. This may not be easy from a given (conformally) flat space endowed with non-Cartesian coordinates. A totally covariant definition of the parameters can be found in \cite{marspeon22}.
\end{remark}

 Via the isomorphism in \eqref{CKVFgeneral} between $\skwend{\mathbb{M}^{1,n+1}}$ and $\mathrm{CKill}(\mathbb{E}^n)$, $F_{can}$ induces a canonical representative $\xi_{can}$ for the orbits $[\xi]$
\begin{equation}
 \xi_{can} = \tilde \xi + \sum_{i=1}^p \mu_i\eta_i,
\end{equation}
 where $\tilde \xi$ and $\eta_i$ are mutually orthogonal CKVs whose form depends on the parity of $n$.

For $n$ even, using coordinates $\{ y^A \} = \{z_1,z_2,\{x_i \}_{i=1}^{2p} \}$, we have:
 \begin{align}
  \tilde \xi &= \frac{1}{2}\Big(\sigma + z_1^2 - z_2^2 - \sum\limits_{i=1}^{2p}x_i^2 \Big) \partial_{z_1} + \lr{\frac{\tau}{2} + z_1 z_2} \partial_{z_2} + z_1 \sum\limits_{i=1}^{2p} x_i \partial_{x_i},\label{CKVFxiapeven}\\
  \eta_i  & = x_{2i-1} \partial_{x_{2i}} - x_{2i} \partial_{x_{2i-1}}.   \label{CAKVFapeven}
 \end{align}
For $n$ odd, using $\{ y^A \} = \{z_1,\{x_i \}_{i=1}^{2p} \}$:
 \begin{align}
   \tilde \xi &= \frac{1}{2}\Big(\sigma + z_1^2  - \sum\limits_{i=1}^{2p}x_i^2 \Big) \partial_{z_1}  + z_1 \sum\limits_{i=1}^{2p} x_i \partial_{x_i},\label{CKVFxiapodd}\\
  \eta_i  & = x_{2i-1} \partial_{x_{2i}} - x_{2i} \partial_{x_{2i-1}}.\label{CAKVFapodd}
 \end{align}

The zero set of these canonical CKVs is  determined by  evaluation of their components.

\begin{lemma}\label{zerosetxican}
 Let $n$ be even. The zero set of $\xi_{can}$, is characterized as follows:
 \begin{enumerate}[a)]
  \item If $\tau \neq 0$, $\xi_{can}$ has exactly two isolated zeros at $z_1 =  \pm\left(\frac{1}{2}(-\sigma + \sqrt{\sigma^2 + \tau^2})\right)^{1/2}$, $z_2 = - \frac{\tau}{2 z_1}$, and $x_i = 0$.
  \item If $\tau = 0$ and $\sigma \leq 0$, $\xi_{can}$ has two isolated zeros at $z_1 =\pm (-\sigma)^{1/2}$ (which merge to one if $\sigma = 0$), with $z_2 = 0$ and $x_i = 0$.
  \item If $\tau = 0$, $\sigma > 0$, and exactly $k$ parameters among $\{\mu_i\}$ are non-zero, the zeros form an $(n-2k-2)$-dimensional submanifold
\begin{align}
 \mathcal{Z}(\xi_{can}) & = \Big\{  z_2^2 + \sum\limits_{i=2k+1}^{2p}x_i^2  ={\sigma}, z_1 = 0 \Big\} \bigcap_{i=1}^{2k} \Big\{x_i = 0 \Big\},
\end{align}
When $k  = p$, then $\zeroset{\xi_{can}}$ contains only two isolated points at $z_1 = 0, z_2 = \pm\aa^{1/2},x_i = 0.$
 \end{enumerate}
\end{lemma}

\begin{proof}
 By orthogonality of $\tilde \xi$ and $\eta_i$, the zero set $\xi_{can}$ satisfies $\zeroset{\xi_{can}} = \zeroset{\tilde \xi} \bigcap_{i=1}^p \zeroset{\mu_i \eta_i}$, thus the Lemma follows by evaluating the vanishing conditions of \eqref{CKVFxiapeven} and \eqref{CAKVFapeven}.

 Clearly, if $\mu_i = 0$ then $\zeroset{\mu_i \eta_i} = \mathbb{R}^n$, hence, we only take into account non-zero $\mu_i$. Assuming that  $\mu_1, \cdots, \mu_k \neq 0$ we get
 \begin{equation}\label{zerosetai}
  \bigcap_{i=1}^p \zeroset{\mu_i \eta_i} =  \bigcap_{i=1}^{2k} \{x_i = 0 \}.
 \end{equation}
Consider first $\tau \neq 0$. The vanishing of the $z_2$-component of $\tilde \xi$,
\begin{equation}\label{z2component}
\tilde \xi^{z_2} = \frac{\tau}{2} + z_1 z_2 = 0,
\end{equation}
prevents both $z_1$ and $z_2$ from becoming zero.
Thus, by the $i$-components of $\tilde \xi$,
\begin{equation}
 \tilde \xi^i = z_1 x_i,
\end{equation}
we have $x_i = 0$ for $i = 1,\cdots, 2p$. Therefore, writing $z_2 = -\tau/(2z_1)$ and inserting it into the $z_1$-component
\begin{equation}
  \tilde \xi ^{z_1}=\frac{1}{2}\Big(\sigma + z_1^2 - z_2^2 - \sum\limits_{i=1}^{2p}x_i^2 \Big)
\end{equation}
yields a biquadratic poloynomial in $z_1$ whose only real roots are the values given in part $a)$ of the Lemma. Note the this case $\zeroset{\tilde \xi} \subseteq \bigcap_{i=1}^p \zeroset{\mu_i \eta_i}$.

For $\tau = 0$, by \eqref{z2component}, either $z_1$ or $z_2$ must vanish. Consider first $\aa <0 $. If $z_1 = 0$,
then $\tilde \xi^i = 0$ for $i = 1 \cdots 2p$, but then $\tilde \xi^{z_1} = 0$ has no real solution. Hence, $z_1\neq 0$, $z_2 = 0$ and by the $i$-components $\tilde \xi^i $, one has $x_i = 0$ for $i = 1,\cdots, 2p$. Then, evaluating $z_1$-component one obtains $z_1 = \pm(-\aa)^{1/2}$. The $\aa = 0$ case follows from an analogous argument, showing part $b)$ of the lemma.

Consider now $\aa >0$. If $z_1 = 0$, then $\tilde \xi^i = 0$ for $i = 1 \cdots 2p$ and $\tilde \xi^{z_1} = 0$ yields $  z_2^2 + \sum\limits_{i=2k+1}^{2p}x_i^2  ={\sigma}$. Intersecting this with the zeros of the $\mu_i \eta_i$ vectors (equation \eqref{zerosetai}) directly gives the submanifold $\mathcal{Z}(\xi_ {can})$ in part $c)$ of the Lemma. On the other hand, if $z_2 = 0$, $i$-components require either $z_1 = 0$ or $x_i =0$, $i = 1,\cdots, 2p$. The latter is incompatible with real solutions of $\tilde \xi^{z_1} = 0$, so it must be $z_1= z_2 = 0$ and $ \sum\limits_{i=2k+1}^{2p}x_i^2  ={\sigma}$. Hence, starting from $z_2= 0$ one just finds a subset of $\mathcal{Z}(\xi_ {can})$. This finishes part $c)$ of the Lemma.

\end{proof}

\begin{corollary}\label{corozerosetxican}
 Let $n$ be odd. The zero set of $\xi_{can}$, is characterized as follows:
 \begin{enumerate}[a)]

  \item If $\sigma \leq 0$, $\xi_{can}$ has two isolated zeros at $z_1 =\pm (-\sigma)^{1/2}$ (which merge to one if $\sigma = 0$), with $z_2 = 0$ and $x_i = 0$.
  \item If $\sigma > 0$, and exactly $k$ parameters among $\{\mu_i\}$ are non-zero, the zeros form an $(n-2k-2)$-dimensional submanifold
\begin{align}
 \mathcal{Z}(\xi_{can}) & = \Big\{  \sum\limits_{i=2k+1}^{2p}x_i^2  ={\sigma}, z_1 = 0 \Big\} \bigcap_{i=1}^{2k} \Big\{x_i = 0 \Big\},
\end{align}
When $k  = p$, then $\zeroset{\xi_{can}} =  \emptyset$.
 \end{enumerate}
\end{corollary}

\begin{proof}
Consider the isometric embedding
 \begin{equation}\label{embeddingapp}
  \iota: \mathbb{E}^n \hookrightarrow \mathbb{E}^{n+1},\qquad \iota(\mathbb{E}^n)= \{z_2=0 \}.
 \end{equation}
 Then $\iota_\star(\xi_{can})$ is an even dimensional canonical CKV with $\tau=0$ restricted to $z_2 = 0$. The Corollary now follows by applying the results of Lemma \ref{zerosetxican} with $\tau = 0$ at $z_2 = 0$.
\end{proof}

\begin{remark}\label{remarkrotation}
To give a more geometric perspective to our results, we single out other representatives of $\xi \in [\xi_{\mathrm{can}}]$ whose associated flows are easy to identify. This turns out to be simpler when expressed in terms of the associated endomorphisms.

We look for a representative $F = \operatorname{Ad}_L(F_{\mathrm{can}})$ with the same block structure as $F_{\mathrm{can}}$. To this end, we restrict to block-diagonal transformations $L$ belonging to the subgroup $O^+(1,3) \times SO(2)^p$ if $n$ is even (resp. $O^+(1,2) \times SO(2)^p$ if $n$ is odd), embedded block-diagonally in $O^+(1,n+1)$. Specifically, consider $L$ of the form $L = L_0 \oplus \mathbb{I}_{2p}$, where $L_0 \in O^+(1,3)$ if $n$ is even (resp. $L_0 \in O^+(1,2)$ if $n$ is odd) and $\mathbb{I}_{2p}$ is a $2p \times 2p$ identity matrix. Under these transformations, determining whether $F$ and $F_{\mathrm{can}}$ belong to the same orbit reduces to analyzing the first $4 \times 4$ block of $F$ if $n$ is even (resp. $3 \times 3$ block if $n$ is odd). Namely, extracting this block from $F$, we can calculate its $\sigma, \tau$ (resp. $\sigma$) parameters according with Definition~\ref{defgammamu} and these must coincide with the values prescribed by $F_{\mathrm{can}}$.

In terms of CKVs, this means that we have a representative $\xi \in [\xi_{\mathrm{can}}]$ given by
\begin{equation}
\xi = \varphi_\star(\tilde\xi) + \sum_{i=1}^{p} \mu_i \varphi_\star(\eta_i), \qquad \varphi \in \operatorname{Conf}(\mathbb{E}^n),
\end{equation}
with $\varphi$ partially fixed so that $\varphi_\star(\eta_i) = \eta_i$, leaving the remaining freedom to acting on $\tilde\xi$. For $n$ even, $\varphi_\star(\tilde\xi)$ can be brought to one of the following forms:
\begin{enumerate}
  \item If $\tau \neq 0$, an orthogonal sum of a homothety with divergence $\nu = \sqrt{(-\aa + \sqrt{\aa^2 + \bb^2})/2}$ and a rotation in the $(z_1, z_2)$-plane with angular velocity $\mu_s = \sqrt{(\aa + \sqrt{\aa^2 + \bb^2})/2}$.

  \item If $\tau = 0$, then:
  \begin{itemize}
    \item if $\sigma > 0$, a rotation in the $(z_1, z_2)$-plane;
    \item if $\sigma = 0$, a translation along the $z_1$-axis with angular velocity $\mu_s = \sqrt{\aa}$;
    \item if $\sigma < 0$, a homothety with divergence $\nu = \sqrt{-\aa}$.
  \end{itemize}
\end{enumerate}
For $n$ odd, the cases $\sigma \le 0$ are identical to the even-dimensional ones with $\tau = 0$. If $\sigma > 0$, $\varphi_\star(\tilde\xi)$ can also be reduced to a rotation, although not while simultaneously keeping all $\varphi_\star(\eta_i) = \eta_i$.
\end{remark}

\section{Calculation of $|dF|_{g_M}^2$ }\label{appdF}

Starting from the factorization
\begin{equation}
  F = |X|_{g_M}^2 = (\Om^2 - \alpha_+) (\Om^2 - \alpha_-),
\end{equation}
where recall
\begin{equation}
\alpha_\pm := (z_1^2 + z_2^2 + \rho^2 - \cc) + 2\beta_\pm,\qquad \beta_\pm := \pm\Big(-\cc z_1^2 - \sum_i^p \mu_i^2 \rho_i^2 - \left(\frac{\tau}{2}\right)^2 - \tau z_1 z_2\Big)^{1/2},
\end{equation}
we determine the causal character of the zero level sets of $F$ by computing its normal vector evaluated at the roots $\Om^2 = \alpha_\pm$. The calculation can be worked out simultaneously for both roots since it is analogous.

Taking the differential of $F$ yields
\begin{equation}
  dF = (\alpha_\pm - \alpha_\mp) (2\Om d\Om - d\alpha_\pm).
\end{equation}

First, we compute  $(\alpha_\pm - \alpha_\mp) d\alpha_\pm$. We start by taking the differential
\begin{equation}
  d\alpha_\pm = 2z_1 dz_1 + 2z_2 dz_2 + 2 \rho d \rho + \frac{2}{2\beta_\pm} \left(-2\cc z_1 dz_1 - \tau z_1 dz_2 - \tau z_2 dz_1 - 2\sum_{i=1}^p \mu_i^2 \rho_i d \rho_ i \right).
\end{equation}
To eliminate the squared roots, at  $\Om^2 = \alpha_\pm = (z_1^2 + z_2^2 + \rho^2 - \cc) + 2\beta_\pm$, we can substitute $2 \beta_{\pm}$ by
\begin{equation}
  2\beta\pm = \Om^2 - z_1^2 - z_2^2 - \rho^2 + \cc
\end{equation}
and find a common denominator
\begin{align}
d \alpha_\pm = & \frac{2}{2\beta_\pm} \Big( z_1 \left(  \Om^2 - z_1^2 - z_2^2 - \rho^2 + \cc - 2\cc  \right) dz_1 + z_2 (\Om^2 - z_1^2 - z_2^2 - \rho^2 + \cc) dz_2  \nonumber \\
 + \, &  \rho (\Om^2 - z_1^2 - z_2^2 - \rho^2 + \cc)  d \rho - 2\sum_{i= 1}^p\mu_i^2 \rho_i d\rho_i - \tau z_1 dz_2 - \tau z_2 dz_1 \Big).\label{dalpha1}
\end{align}
Denoting $g_M^\sharp$ the contravariant metric defined by $g_M$, observe that $g^\sharp_M(d \rho_i,d \rho_j) = \delta_{ij}$ and taking also into account that $\rho d \rho = \sum_i^p \rho_i d\rho_i$, it follows
\begin{equation}
 g_M^\sharp(\rho d \rho, \rho_i d \rho_i) = \rho_i^2.\label{gdrhodrhoi}
\end{equation}
Also noting that $(\alpha_\pm - \alpha_\mp) = 4 \beta_\pm$, by \eqref{dalpha1} and \eqref{gdrhodrhoi},  we have
\begin{align}
 (\alpha_\pm - \alpha_\mp)^2 |d\alpha_\pm|_{g_M}^2 =  16  \left( z_1^2 (\Om^2 - z_1^2 - z_2^2 - \rho^2 + \cc - 2\cc)^2 + z_2^2 (\Om^2 - z_1^2 - z_2^2 - \rho^2 + \cc)^2 \right. \nonumber \\
 +\rho^2  (\Om^2 - z_1^2 - z_2^2 - \rho^2 + \cc)^2 + 4 \sum_{i=1}^p\left(\mu_i^4 \rho_i^2 -\mu_i^2 \rho_i^2(\Om^2 - z_1^2 - z_2^2 - \rho^2 + \cc)\right)
 \\
  + \tau^2 z_2^2 - 2\tau z_1 z_2 (\Om^2 - z_1^2 - z_2^2 - \rho^2 + \cc)
 \left. + \, \tau^2 z_1^2 - 2\tau z_1 z_2 (\Om^2 - z_1^2 - z_2^2 - \rho^2 + \cc - 2\cc) \right). \label{normdalpha+}
\end{align}
We now simplify expresion \eqref{normdalpha+}. By substituting the identity
\begin{equation}
 \left(\Om^2 - z_1^2 - z_2^2 - \rho^2 + \cc\right)^2 = 4 \beta_\pm^2 = 4 \left(-\cc z_1^2 - \sum_{i=1}^p\mu_i^2 \rho_i^2 - \frac{\tau^2}{4} - \tau z_1 z_2\right),
\end{equation}
and suitably cancelling, the first term in the first line in \eqref{normdalpha+} simplifies to
\begin{align*}
 & z_1^2 \left(\Om^2 - z_1^2 - z_2^2 - \rho^2 + \cc - 2\cc\right)^2 \\
 & = z_1^2 \left((\Om^2 - z_1^2 - z_2^2 - \rho^2 + \cc)^2 +4 \cc^2 - 4\cc (\Om^2 - z_1^2 - z_2^2 - \rho^2 + \cc) \right) \\
 &= 4z_1^2 \left( - \sum_{i=1}^p\mu_i^2 \rho_i^2 - \frac{\tau^2}{4} - \tau z_1 z_2\right) - 4\cc \Om^2 z_1^2  + 4\cc z_1^2 z_2^2 + 4\cc z_1^2 \rho^2.
\end{align*}
 and, similarly, the second one becomes
\begin{align*}
 & z_2^2 \left(\Om^2 - z_1^2 - z_2^2 - \rho^2 + \cc\right)^2 = 4z_2^2 \left(-\cc z_1^2 - \sum_{i=1}^p\mu_i^2 \rho_i^2 - \frac{\tau^2}{4} - \tau z_1 z_2\right).
\end{align*}
The same identity applied to the second line in \eqref{normdalpha+} yields
\begin{align}
  & \,\rho^2  (\Om^2 - z_1^2 - z_2^2 - \rho^2 + \cc)^2 + 4 \sum_{i=1}^p\left(\mu_i^4 \rho_i^2 -\mu_i^2 \rho_i^2(\Om^2 - z_1^2 - z_2^2 - \rho^2 + \cc)\right) \\
%
  =  & \, 4\rho^2 \left(-\cc z_1^2 -\sum_{i=1}^p \mu_i^2 \rho_i^2  - \frac{\tau^2}{4} - \tau z_1 z_2\right) + 4 \sum_{i=1}^p\left(\mu_i^4 \rho_i^2 -\mu_i^2 \rho_i^2(\Om^2 - z_1^2 - z_2^2 - \rho^2 + \cc)\right)\\
   =  & \, 4\rho^2 \left(-\cc z_1^2 - \frac{\tau^2}{4} - \tau z_1 z_2\right) + 4 \sum_{i=1}^p\left(\mu_i^4 \rho_i^2 -\mu_i^2 \rho_i^2(\Om^2 - z_1^2 - z_2^2  + \cc)\right)
\end{align}
And finally, the third line immediately reduces to
\begin{align*}
 & \tau^2 z_1^2 + \tau^2 z_2^2 - 2\tau z_1 z_2 \left(\Om^2 - z_1^2 - z_2^2 - \rho^2 + \cc\right) - 2\tau z_1 z_2 \left(\Om^2 - z_1^2 - z_2^2 - \rho^2 + \cc - 2\cc\right) \\
 & = \,\tau^2 z_1^2 + \tau^2 z_2^2 - 4\tau z_1 z_2 \left(\Om^2 - z_1^2 - z_2^2 - \rho^2 \right).
\end{align*}
Substituting these expressions back into \eqref{normdalpha+} and cancelling terms, we arrive at
\begin{align}
 (\alpha_\pm - \alpha_\mp)^2 |d\alpha_\pm|^2 & = 16 \left(-4 \cc \Om^2 z_1^2-\rho^2 \tau^2-4 \tau \Om^2 z_1 z_2+4\sum_{i=1}^p \big(-\Om^2\mu^2_ i \rho_i^2 + \mu_i^2 \rho_i^2(\mu_i^ 2 - \cc)\big)\right). \label{normdalpha+2}
\end{align}

On the other hand, evaluating the $d \Om$  component gives
\begin{align}\label{eqdOm}
 (\alpha_\pm-\alpha_\mp)^2 4 \Om^2|d \Om|_{g_M}^2&= -64\Om^2 \left(-\cc z_1^2 - \sum_{i=1}^p\mu_i^2 \rho_i^2 - \frac{\tau^2}{4} - \tau z_1 z_2\right),
\end{align}
where we have just applied the identity   $(\alpha_\pm - \alpha_\mp) = 4 \beta_\pm$.

Combining \eqref{normdalpha+2} and \eqref{eqdOm}, the squared norm of the normal vector is
\begin{align*}
{|dF|_{g_M}^2} &= (\alpha_\pm - \alpha_\mp)^2(4 \Om^2 |d\Om|_{g_M}^2 + |d \alpha_\pm|^2) = 64 \left( \frac{\tau^2}{4}(\Om^2 - \rho^2)  +  \sum_{i=1}^p\mu_i^2 \rho_i^2(\mu_i^ 2 - \cc) \right).
\end{align*}
Substituting $\Om^2 = \alpha_\pm$ and inserting the definition of $\alpha_\pm$ yields the final expression
\begin{align*}
{|dF|_{g_M}^2}  = 64 \left( \frac{\tau^2}{4}(z_1^2 +z_2^2 - \aa +2 \beta_\pm)  +  \sum_{i=1}^p\mu_i^2 \rho_i^2(\mu_i^ 2 - \cc) \right).
\end{align*}

\end{document}